\documentclass[journal,onecolumn]{IEEEtran}
\usepackage{amsmath,amsfonts, amssymb}
\usepackage{bbm}
\usepackage{algorithmic}
\usepackage{algorithm}
\usepackage{array}
\usepackage{subcaption}
\usepackage{textcomp}
\usepackage{stfloats}
\usepackage{hyperref}
\usepackage{url}
\usepackage{xcolor}
\usepackage{physics}
\usepackage{verbatim}
\usepackage{graphicx}
\usepackage{cite}

\newcommand{\bea}{\begin{eqnarray}}
\newcommand{\eea}{\end{eqnarray}}

\newtheorem{theorem}{Theorem}
\newtheorem{corollary}[theorem]{Corollary}
\newtheorem{lemma}[theorem]{Lemma}
\newtheorem{proposition}[theorem]{Proposition}
\newtheorem{remark}{Remark}

\newtheorem{definition}{Definition}

\begin{document}

\title{Optimal Privacy-Utility Trade-Offs in LDP: Functional and Geometric Perspectives}

\author{Seung-Hyun Nam,~\IEEEmembership{Member,~IEEE}, Hyun-Young Park,~\IEEEmembership{Member,~IEEE},
\\ Si-Hyeon Lee,~\IEEEmembership{Senior Member,~IEEE}
\thanks{S.-H. Nam is with the Faculty of Information Technology, University of Jyväskylä, Jyväskylä 40100, Finland (e-mail: seunghyun.sh.nam@jyu.fi).}
\thanks{H.-Y. Park and S.-H. Lee are with the School of Electrical Engineering, KAIST, Daejeon 34141, South Korea  (e-mail: phy811@kaist.ac.kr; sihyeon@kaist.ac.kr). (Corresponding Author: Si-Hyeon Lee)}
}

\maketitle

\begin{abstract}
Local differential privacy (LDP) has emerged as a gold-standard framework for privacy-preserving data analysis.
However, characterizing the optimal privacy–utility trade-off (PUT) and the corresponding optimal LDP channels remains largely fragmented, relying on problem-specific, case-by-case analyses.
In this work, we develop a unified theoretical framework that systematically characterizes the optimal PUT and optimal LDP channels for general privacy-preserving statistical decision-making problems.
We first identify key functional properties of Bayesian and minimax risks as functions of the LDP channel, including the data processing inequality (DPI), direct-sum quasi-convexity (or additivity), concavity, and symmetry invariance.
Leveraging these properties, we reduce the optimization domain required to compute the optimal PUT.
To this end, we reinterpret the existing result implying the one-to-one correspondence between maximal LDP channels under the Blackwell order and a finite-dimensional polytope, through an alternative convex geometric derivation.
This result renders the optimal PUT computationally tractable via vertex enumeration or linear programming.
Furthermore, when the underlying problem exhibits symmetries characterized by a transitive group action, we show that the optimal PUT can be computed by comparing finitely many values, rather than optimizing over a high dimensional polytope.
Our framework applies broadly beyond risk minimization, encompassing the maximization of information-theoretic measures.
We demonstrate the efficacy of our theoretical framework by recovering or strengthening several known results, and deriving exact analytic expressions for the optimal PUTs in specific tasks that were previously unaddressed.
\end{abstract}

\section{Introduction}\label{sec:intro}
While data-driven statistical inference has enabled a wide range of applications, the collection and analysis of user-level data continue to raise critical privacy concerns \cite{narayananRobustDeanonymizationLarge2008a, fredriksonModelInversionAttacks2015, dworkExposedSurveyAttacks2017, geipingInvertingGradientsHow2020}.
In settings where data providers do not fully trust a central data collector, local differential privacy (LDP) has emerged as a gold-standard mathematical framework for privacy-preserving data analysis \cite{kasiviswanathanWhatCanWe2011a, dworkAlgorithmicFoundationsDifferential2013, duchiLocalPrivacyStatistical2013a}.
Under LDP, each individual randomizes their data via a stochastic channel--commonly referred to as an LDP mechanism or LDP channel--prior to transmission.
This ensures rigorous privacy guarantees without relying on a trusted third party or assumptions about an adversary's computational power \cite[Thm.~14]{issaOperationalApproachInformation2020}.

However, this privacy protection inevitably comes at a cost.
By the data processing inequality (DPI), the randomized data necessarily contains less information than the original data, which degrades the utility of downstream inference tasks.
Consequently, a central problem in privacy-preserving data analysis is to characterize the optimal privacy–utility trade-off (PUT) and to design optimal LDP channels that achieve this optimal PUT.
Despite extensive research, many existing studies have characterized optimal PUTs and optimal LDP channels on a case-by-case basis for each utility under consideration, often relying on problem-specific and ad hoc techniques.

On the other hand, several works have made an effort to provide a comprehensive and unified understanding of optimality across various scenarios \cite{chaiCriterionPrivacyProtection2018a, ghoshUniversallyUtilitymaximizingPrivacy2009, fernandesUniversalOptimalityRobust2023b, gupteUniversallyOptimalPrivacy2010, brennerImpossibilityDifferentiallyPrivate2014a, kairouzExtremalMechanismsLocal2016b}.
Motivated by these studies, we propose a unified approach characterizing optimality based on functional properties of utility, and geometric properties of the set of LDP channels.
The functional properties considered in this paper naturally arise in privacy-preserving statistical decision making scenarios, which include two representative applications: hypothesis testing and parametric estimation.
Furthermore, we present a method for exploiting the underlying symmetries of the scenario to characterize the optimality more efficiently.
We discuss the related works and our contributions in further detail below.

\begin{figure*}[h]
    \centering
    \includegraphics[width=0.9\linewidth]{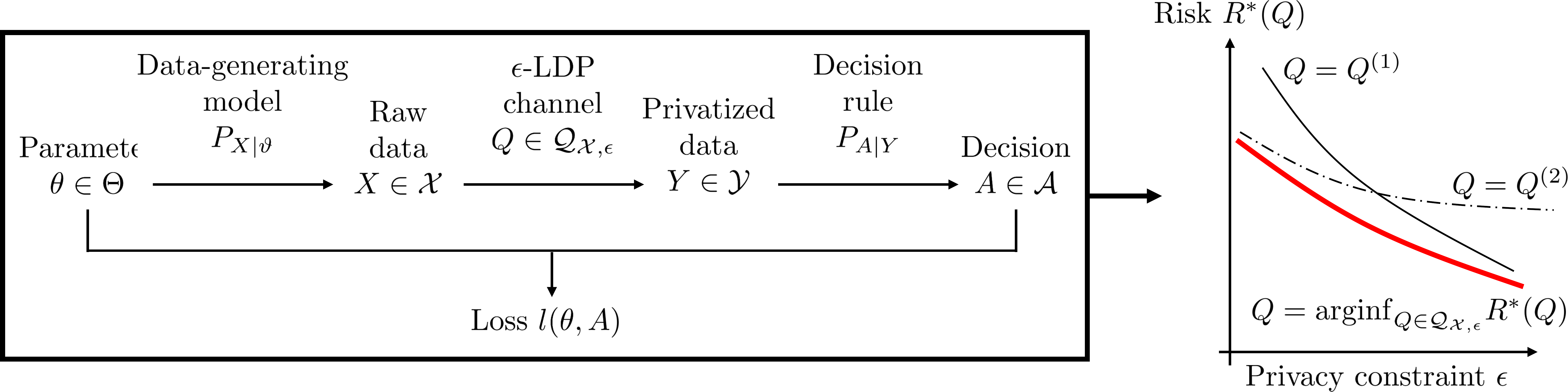}
    \caption{Privacy-preserving statistical decision-making scenario and the optimal PUT (details are in Section~\ref{sec:stat_dec}).}
    \label{fig:priv_dec}
\end{figure*}

\begin{figure*}[h]
    \centering
    \includegraphics[width=0.8\linewidth]{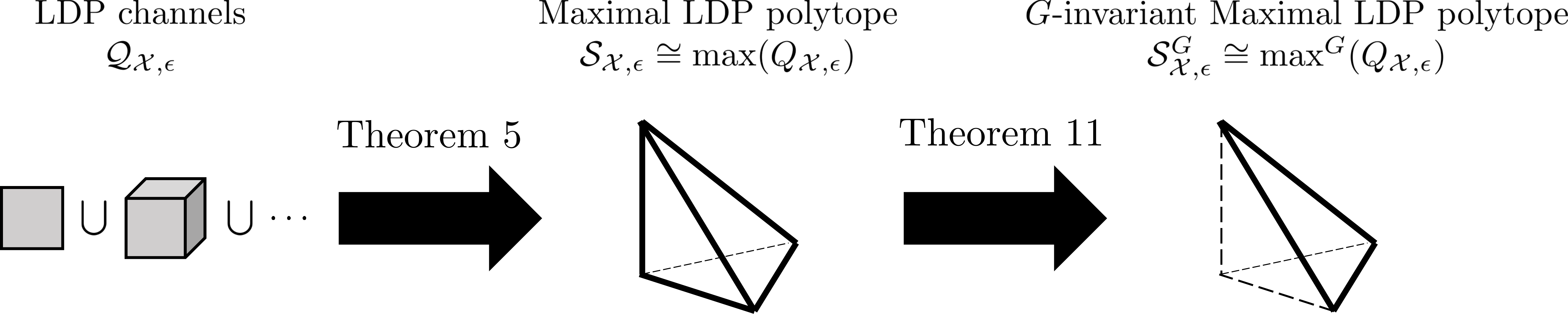}
    \caption{A conceptual visualization of the dimensionality reduction of the optimization domain for computing the optimal PUT (details are in Sections~\ref{sec:max_LDP} and~\ref{sec:G_inv}).}
    \label{fig:dim_red}
\end{figure*}

\subsection{Related Work}\label{sec:related_work}
The framework of LDP originally emerged as a local variant of differential privacy (DP) \cite{dworkDifferentialPrivacy2006}, where DP assumes a trusted central data collector, whereas LDP does not.
For both DP- and LDP-based privacy-preserving tasks, numerous prior works have characterized the optimal PUTs and the corresponding optimal (L)DP channels \cite{ghoshUniversallyUtilitymaximizingPrivacy2009, duchiLocalPrivacyStatistical2013a, gengOptimalNoiseAddingMechanism2016a, kairouzExtremalMechanismsLocal2016b, yeOptimalSchemesDiscrete2018, acharyaHadamardResponseEstimating2019, wangLocalDifferentialPrivate2019, chenBreakingCommunicationPrivacyAccuracyTrilemma2020, barnesFisherInformationLocal2020, asiOptimalAlgorithmsMean2022a, parkExactlyOptimalCommunicationEfficient2024, namAchievingExactlyOptimal2024, namOptimalPrivateDiscrete2024b, parkExactlyMinimaxOptimalLocally2024, kalininEfficientEstimationGaussian2025, yoonFundamentalLimitDiscrete2025a, gentleNecessityBlockDesigns2025a, namQuantumAdvantagePrivate2025, amorinoRoleSymmetryStaircase2026, melbourneOptimalityStaircaseMechanisms2026,chaiCriterionPrivacyProtection2018a, fernandesUniversalOptimalityRobust2023b, gupteUniversallyOptimalPrivacy2010, brennerImpossibilityDifferentiallyPrivate2014a}.
While the literature spans a wide range of specialized settings and generalizations, this work focuses on non-interactive $\epsilon$-LDP channels between finite input and output alphabets.
Moreover, we adopt the notion of exact optimality that does not rely on approximation terms, thereby providing a strictly stronger guarantee than the commonly used notion of order optimality.

Within this specific setting (non-interactive $\epsilon$-LDP channel between finite input and output alphabets), the extremal LDP mechanisms proposed in \cite{kairouzExtremalMechanismsLocal2016b} achieve the optimal privacy--utility trade-off (PUT) for several tasks, although existing proofs often rely on ad hoc technical arguments.
In particular, \cite{kairouzExtremalMechanismsLocal2016b} showed that \emph{extremal mechanisms} are optimal for a class of utility functions satisfying certain properties, including all $f$-divergences.
Moreover, the \emph{subset selection} (SS) mechanism, a notable special case of an extremal mechanism, has been shown to be optimal in several settings.
For instance, it maximizes mutual information under a uniform input distribution \cite{wangLocalDifferentialPrivate2019}, achieves asymptotic minimax optimality in discrete distribution estimation \cite{yeOptimalSchemesDiscrete2018}, and is optimal for asymptotic multiple and binary composite hypothesis testing under certain settings \cite{namQuantumAdvantagePrivate2025}.

Beyond domain-specific approaches, several works have tried to give a comprehensive understanding on the optimal PUT and the optimal (L)DP channels \cite{chaiCriterionPrivacyProtection2018a, ghoshUniversallyUtilitymaximizingPrivacy2009, fernandesUniversalOptimalityRobust2023b,
gupteUniversallyOptimalPrivacy2010}.
Considering DP, \cite{ghoshUniversallyUtilitymaximizingPrivacy2009} showed that a \emph{geometric mechanism} is universally optimal for the class of counting query mechanisms and monotone loss functions under Bayesian framework, and the following work \cite{gupteUniversallyOptimalPrivacy2010} extended this result to minimax framework.
Also, \cite{fernandesUniversalOptimalityRobust2023b} analyzed universal optimality in Bayesian setup based on the concept of \emph{quantitative information flow} under a more general constraint known as \emph{metric DP} \cite{chatzikokolakisBroadeningScopeDifferential2013}.

Regarding LDP, \cite{kairouzExtremalMechanismsLocal2016b} derived the optimality of extremal mechanisms from a geometric perspective, by considering the set of fixed-dimensional LDP channels as a polytope and exploiting the properties of its vertices.
Building on this viewpoint, \cite{holohanExtremePointsLocal2017} analyzed the full set of vertices, revealing the existence of nontrivial vertices that do not correspond to extremal mechanisms.
This suggests that a more refined geometric framework is needed to rigorously establish the optimality of extremal mechanisms.

Shifting focus to the functional properties of utility, \cite{kairouzExtremalMechanismsLocal2016b} considered a class of utility functions that can be decomposed into sublinear components, which includes all $f$-divergences.
However, this class remains limited; for example, the minimax risk in privacy-preserving decision-making considered in this work does not fall into this class.\footnote{Minimizing risk is equivalent to maximizing utility.}
Notably, \cite[Sec.~4]{chaiCriterionPrivacyProtection2018a} showed that a subclass of extremal mechanisms maximizes the utility whenever the utility satisfies the data processing inequality (DPI).
A similar result was also reported in \cite{fernandesUniversalOptimalityRobust2023b} under the name of \emph{vertex mechanisms}.
This reduction from the DPI is useful in the sense that it reduces the domain of optimization for computing the optimal PUT from an infinite dimensional set to a finite dimensional set.
However, the limitation is that this subclass is still a high-dimensional set, whose dimension increases exponentially with the size of the input alphabet.
Also, the DPI itself does not offer further guidance on how to actually compute the optimal PUT.
Hence, further functional properties of utilities need to be investigated for efficient computation.

Another line of work relies on \emph{contraction coefficients}, which are closely linked to the strong DPI, to analyze the optimal PUT \cite{duchiLocalPrivacyStatistical2013a, asoodehLocalDifferentialPrivacy2021}.
However, this approach often fails to yield an exact characterization and still requires problem-specific techniques to derive tight contraction coefficients for different utility measures. 
Among more recent approaches, \cite{amorinoRoleSymmetryStaircase2026} studied the role of symmetry in maximizing Fisher information over LDP channels, while \cite{melbourneOptimalityStaircaseMechanisms2026} established the optimality of the \emph{staircase mechanism} \cite{gengOptimalNoiseAddingMechanism2016a} in minimizing expected cost among additive global DP channels, leveraging the monotonicity of the expected cost under a certain symmetric operation on a DP channel.

\subsection{Our Contributions}
We develop a unified theoretical framework for characterizing the optimal PUT and the corresponding optimal LDP channels.
Compared to the most relevant works \cite{chaiCriterionPrivacyProtection2018a, fernandesUniversalOptimalityRobust2023b}, we further investigate the functional properties of utilities beyond the DPI, which naturally arise from a general privacy-preserving statistical decision-making scenario or from the maximization of information measures.
While a Bayesian decision-making scenario was considered in \cite{fernandesUniversalOptimalityRobust2023b}, our approach also encompasses the minimax scenario.
In addition, we show that exploiting symmetry can significantly simplify the problem of characterizing the optimal PUT and optimal LDP channels, which was not explored in those works.
In more detail, our main contributions are summarized as follows.

1) We introduce a general framework for privacy-preserving statistical decision-making, illustrated in Fig.~\ref{fig:priv_dec}, which covers hypothesis testing and parametric estimation as special cases.
In Proposition~\ref{prop:BM_structure}, we show that the optimal Bayesian and minimax risks, viewed as functions $R^*(Q)$ of an LDP channel $Q$, satisfy the following functional properties: the data processing inequality (DPI), direct-sum quasi-convexity (or additivity), concavity, and group invariance.
Leveraging only these properties, Theorem~\ref{thm:opt_PUT_red} reduces the optimization domain required to compute the optimal PUT, $\inf_{Q} R^*(Q)$.
As a result, the framework extends beyond risk minimization to also incorporate the maximization of information measures, such as mutual information, $f$-divergences, and Fisher information over LDP channels \cite{kairouzExtremalMechanismsLocal2016b, wangLocalDifferentialPrivate2019, barnesFisherInformationLocal2020}.

2) Although Theorem~\ref{thm:opt_PUT_red} reduces the optimization domain, the problem remains computationally intractable since it still includes LDP channels with arbitrarily large (finite) output alphabets.
Regarding this point, the results in \cite{chaiCriterionPrivacyProtection2018a, fernandesUniversalOptimalityRobust2023b} imply that the domain can be reduced to a finite dimensional subset of extremal mechanisms whenever the utility satisfies the DPI.
In this work, we reinterpret this result through a convex geometric derivation in Theorem~\ref{thm:extremal_LDP_is_max}, which establishes a one-to-one correspondence between a finite dimensional polytope and the set of maximal LDP channels under the Blackwell order \cite{blackwellComparisonExperiments1951}, as depicted in Figure~\ref{fig:dim_red}.
As a consequence, the optimal PUT can be computed via vertex enumeration or linear programming over this polytope when $R^*(Q)$ is concave or direct-sum additive in $Q$, respectively (cf. Remark~\ref{rmk:LP_DPI}, \cite[Thm.~4]{kairouzExtremalMechanismsLocal2016b}).
Compared to \cite{chaiCriterionPrivacyProtection2018a}, our derivation shows how this approach can be generalized to more general privacy constraints, such as metric DP \cite{chatzikokolakisBroadeningScopeDifferential2013}, which was also considered in \cite{fernandesUniversalOptimalityRobust2023b} (Proposition~\ref{prop:max_LDP_geo}, Remark~\ref{rmk:general_const}).
Compared to \cite{fernandesUniversalOptimalityRobust2023b}, we directly view the likelihoods of an LDP channel as an element inside a cone, rather than introducing the notion of \emph{hyper distributions}.

3) While Theorem~\ref{thm:extremal_LDP_is_max} may be useful, the dimension of the polytope increases exponentially with the size of the input alphabet, making the computation of the optimal PUT still challenging.
Moving beyond this limitation, Theorem~\ref{thm:G_inv_poly} establishes a one-to-one correspondence between group-invariant maximal LDP channels and a lower-dimensional polytope, further reducing the optimization domain when $R^*$ is invariant under group symmetries.
Notably, when $R^*$ is invariant under the permutation symmetry, this reduction shrinks the domain to a finite set of SS mechanisms \cite{yeOptimalSchemesDiscrete2018} whose cardinality is one less than that of the input alphabet (Remark 6).
Consequently, the optimal PUT can be computed by evaluating only a small number of candidates, rather than optimizing over a high-dimensional polytope.
Without being restricted to permutation symmetry, Corollary~\ref{cor:PUT_transitive} reduces the optimization to a finite minimization when $R^*$ is concave and invariant under the induced action of a group acting transitively on the input alphabet.
In Section~\ref{sec:app}, we apply this result to certain one-shot hypothesis testing and parametric estimation problems, and derive exact analytical expressions for the optimal PUTs that were previously unresolved.

\subsection{Organization of the Paper}
The rest of the paper is organized as follows. In Section~\ref{sec:pre}, we introduce the preliminaries on convex geometry and privacy-preserving statistical decision-making scenarios under the LDP constraint.
Section~\ref{sec:main} presents our main results: functional properties of optimal Bayesian and minimax risks, geometric characterization of maximal LDP channels, and domain reduction via symmetry.
We provide applications of these results in Section~\ref{sec:app}, where the optimal PUTs are characterized analytically for two concrete examples.
Section~\ref{sec:conc} concludes the paper by introducing some future directions.
Appendix~\ref{sec:pfs} contains the remaining proofs.

\section{Preliminaries}\label{sec:pre}

\subsection{Notations}\label{sec:notations}
For a positive integer $m \in \mathbb{N}$, we define $[m] := \{1, 2, \dots, m\}$.
For a finite set $\mathcal{X}$ and a vector $v \in \mathbb{R}^\mathcal{X}$, the support of $v$ is defined as the set of indices corresponding to non-zero components, $\mathrm{supp}(v) := \{x \in \mathcal{X} \mid v_x \neq 0\}$.
The all-zero and all-one vectors in $\mathbb{R}^\mathcal{X}$ are denoted by $\mathbf{0}_{\mathcal{X}}$ and $\mathbf{1}_{\mathcal{X}}$, respectively.
For vectors $u, v \in \mathbb{R}^\mathcal{X}$, the component-wise partial order is denoted by $u \geq v$, meaning $u_x \geq v_x$ for all $x \in \mathcal{X}$. In particular, $u \geq 0$ means $u \geq \mathbf{0}_{\mathcal{X}}$.

Let $\mathcal{X}$ and $\mathcal{Y}$ be finite sets (alphabets).
The set of real-valued matrices indexed by $\mathcal{Y} \times \mathcal{X}$ is denoted by $\mathbb{R}^{\mathcal{Y}\times \mathcal{X}}$.
We denote the $x$-th column, $y$-th row, and $(y,x)$-th component of $A \in \mathbb{R}^{\mathcal{Y}\times \mathcal{X}}$ by $A_{*,x}$, $A_{y,*}$, and $A_{y,x}$, respectively.
We define the class of all matrices with columns indexed by $\mathcal{X}$ and rows indexed by an arbitrary finite alphabet as $\mathbb{R}^{* \times \mathcal{X}}$.\footnote{We use the term ``class'' instead of ``set'' because the aggregate of all finite sets forms a proper class.}
The all-zero and all-one matrices in $\mathbb{R}^{\mathcal{Y}\times \mathcal{X}}$ are denoted by $\mathbf{0}_{\mathcal{Y},\mathcal{X}}$ and $\mathbf{1}_{\mathcal{Y},\mathcal{X}}$, respectively.
For two matrices $A_1 \in \mathbb{R}^{\mathcal{Y}_1 \times \mathcal{X}}$ and $A_2 \in \mathbb{R}^{\mathcal{Y}_2 \times \mathcal{X}}$, we identify their direct-sum $A_1 \oplus A_2 \in \mathbb{R}^{(\mathcal{Y}_1 \sqcup \mathcal{Y}_2) \times \mathcal{X}}$ by stacking them vertically, i.e., 
\begin{equation}
    A_1 \oplus A_2 := \begin{bmatrix} A_1 \\ A_2 \end{bmatrix}.
\end{equation}
Consequently, any matrix $A \in \mathbb{R}^{\mathcal{Y} \times \mathcal{X}}$ constructed from row vectors $\{a_y^\top\}_{y \in \mathcal{Y}}$ can be compactly written as the direct-sum $A = \bigoplus_{y \in \mathcal{Y}} a_y^\top$.

Let $\mathcal{X}$ and $\mathcal{Y}$ be measurable spaces.
The set of all probability measures on $\mathcal{X}$ is denoted by $\mathcal{P}(\mathcal{X})$.
A channel from $\mathcal{X}$ to $\mathcal{Y}$ is a conditional probability (Markov kernel) which is typically denoted by $P_{Y|X}$.
We denote the set of all channels from $\mathcal{X}$ to $\mathcal{Y}$ by $\mathcal{P}(\mathcal{X} \to \mathcal{Y})$.
When $\mathcal{X}$ is a finite set, a probability measure $P \in \mathcal{P}(\mathcal{X})$ is identified with a probability mass vector in $\mathbb{R}_+^{\mathcal{X}}$.
Similarly, when both $\mathcal{X}$ and $\mathcal{Y}$ are finite, a channel $P_{Y|X} \in \mathcal{P}(\mathcal{X} \to \mathcal{Y})$ is identified with a column-stochastic matrix $Q \in \mathbb{R}_+^{\mathcal{Y} \times \mathcal{X}}$, where the matrix element $Q_{y,x} := P_{Y|X}(y|x)$.
We define the class of all channels from a finite $\mathcal{X}$ to an arbitrary finite alphabet as $\mathcal{P}(\mathcal{X} \to *)$.

Let $G$ be a group with identity element $e$.
A (left) group action of $G$ on a set $\mathcal{X}$ is formally defined as a map $\sigma_{(\cdot)}: G \rightarrow ( \mathcal{X} \to \mathcal{X})$ that satisfies $\sigma_e(x) = x$, $\sigma_g (\sigma_h (x)) = \sigma_{gh}(x)$ for all $g, h \in G$ and $x \in \mathcal{X}$.
When a specific group action is fixed and can be understood from the context without ambiguity, we omit the map $\sigma$ and simply write $gx := \sigma_g (x)$ for brevity.
For any element $x \in \mathcal{X}$, its orbit under the action of $G$ is defined as $Gx := \{gx \mid g \in G\}$.
Since the orbits form a partition of $\mathcal{X}$, we denote the set of all such orbits by $\mathcal{X}/G$.

\subsection{Convex Geometry}\label{sec:conv_geo}

We characterize maximal LDP channels using convex-geometric approaches.
For a comprehensive treatment of the concepts in this subsection, we refer the reader to standard texts \cite{bertsimasIntroductionLinearOptimization1997,zieglerLecturesPolytopes1995, rockafellarConvexAnalysis1997}.
While the following concepts can be described in a more general manner, we introduce them in a concise way that is sufficient for our purposes.

\begin{definition}[Polyhedral cone]
    Let $T = \bigoplus_{y \in \mathcal{Y}} t_y^\top \in \mathbb{R}^{\mathcal{Y} \times \mathcal{X}}$ be a matrix constructed from vectors $t_y \in \mathbb{R}^\mathcal{X}$.
    The set 
    \begin{equation}
    \mathcal{K}(T) := \left\{v \in \mathbb{R}^\mathcal{X} \mid Tv \geq 0\right\}
    \end{equation}
    is called a polyhedral cone defined by a constraint matrix $T$.
    We say a polyhedral cone $\mathcal{K}$ is
    \begin{itemize}
        \item trivial if $\mathcal{K} = \{0\}$.
        \item pointed if $v \in \mathcal{K}$ and $-v \in \mathcal{K}$ implies $v = 0$.
    \end{itemize}
\end{definition}
Note that every polyhedral cone is a convex set.

\begin{definition}[Extreme ray \& direction]\label{def:ext_ray}
    Let $v \in \mathbb{R}^\mathcal{X} \setminus \{0\}$.
    The set $\mathcal{R}(v) := \{c v \mid c \in \mathbb{R}_+\}$ is called a ray generated by $v$.
    A ray $\mathcal{R}(v)$ contained in a cone $\mathcal{K}$ is called an extreme ray of $\mathcal{K}$ if
    \begin{equation}
        \forall u,w \in \mathcal{K}, \quad u + w \in \mathcal{R}(v) \Rightarrow u,w \in \mathcal{R}(v).
    \end{equation}
    A vector $v \in \mathcal{K}$ is called an extreme direction of $\mathcal{K}$ if it generates an extreme ray of $\mathcal{K}$.
\end{definition}

\begin{proposition}[{\cite[Cor.~2.6]{bertsimasIntroductionLinearOptimization1997}}]\label{prop:carath_cone}
    A pointed nontrivial polyhedral cone has a finite number of distinct extreme rays.
    Every element in such a cone is a finite conic combination of its extreme directions.
\end{proposition}

\begin{definition}[Active constraint]
    Let $\mathcal{K}(T)$ be a polyhedral cone defined by a constraint matrix $T = \bigoplus_{y \in \mathcal{Y}} t_y^\top \in \mathbb{R}^{\mathcal{Y} \times \mathcal{X}}$, and let $v \in \mathcal{K}(T)$.
    A constraint $t_y$ is active for $v$ if $\langle t_y, v \rangle = 0$.
    Defining the set of active indices as $\mathcal{J}(v) := \{y \in \mathcal{Y} \mid \langle t_y, v \rangle = 0\}$, the active constraint matrix $T^{(v)}$ is the submatrix of $T$ given by
    \begin{equation}
        T^{(v)} := \bigoplus_{y \in \mathcal{J}(v)} t_y^\top.
    \end{equation}
\end{definition}

\begin{proposition}[{\cite[Prop.~2.16]{bertsimasIntroductionLinearOptimization1997}}]\label{prop:ex_ray_rank}
    Let a pointed nontrivial polyhedral cone $\mathcal{K}(T) \subset \mathbb{R}^\mathcal{X}$ be given.
    A non-zero element $v \in \mathcal{K}(T)$ is an extreme direction of $\mathcal{K}(T)$ if and only if
    \begin{equation}
        \ker(T^{(v)}) = \mathrm{span}(\{v\}).
    \end{equation}
\end{proposition}

\subsection{Privacy-Preserving Statistical Decision-Making}\label{sec:stat_dec}
A statistical decision-making scenario is described by a tuple $(\Theta, \mathcal{X}, P_{X|\vartheta}, \mathcal{A}, l)$, where $\Theta$ denotes a parameter space, $\mathcal{X}$ is a raw data space, $P_{X|\vartheta} \in \mathcal{P}(\Theta \rightarrow \mathcal{X})$ is a data-generating model, $\mathcal{A}$ is a set of all possible decisions, and $l:\Theta \times \mathcal{A} \rightarrow \mathbb{R}$ is a loss function \cite{fergusonMathematicalStatisticsDecision1967, bergerStatisticalDecisionTheory1985}.
A privacy-preserving statistical decision-making problem is an instance of a statistical decision-making where raw data $X \in \mathcal{X}$ should be randomized into $Y \in \mathcal{Y}$ through an $\epsilon$-LDP channel $Q \in \mathcal{P}(\mathcal{X}\rightarrow \mathcal{Y})$, and the decision maker only has access to the randomized data $Y$.
In this work, we focus on the case where both $\mathcal{X}$ and $\mathcal{Y}$ are finite, and $\mathcal{Y}$ can be chosen to be an arbitrary finite set.
The formal definition of $\epsilon$-LDP is given as follows.
\begin{definition}[Local differential privacy (LDP)]
    For $\epsilon > 0$, a channel $Q \in \mathcal{P}(\mathcal{X} \rightarrow *)$ is an $\epsilon$-LDP channel if
    \begin{equation}
        \forall x,x' \in \mathcal{X}, \; e^\epsilon Q_{*,x} - Q_{*,x'} \geq 0.
    \end{equation}
    The class of all $\epsilon$-LDP channels from $\mathcal{X}$ is denoted by $\mathcal{Q}_{\mathcal{X},\epsilon}$.
\end{definition}

Based on observing $Y$, the decision maker makes its decision $A \in \mathcal{A}$ through a decision rule $P_{A|Y} \in \mathcal{P}(\mathcal{Y} \rightarrow \mathcal{A})$.
The quality of the decision is measured by a loss function $l:\Theta \times \mathcal{A} \rightarrow \mathbb{R}$.
Since the system is stochastic, it is natural to consider expected loss, called the risk, defined by
\begin{align}
    R&(\theta,Q,P_{A|Y}) := \mathbb{E}\left[ l (\theta, A) \right]
    \\& = \sum\limits_{x \in \mathcal{X}, y \in \mathcal{Y}}  P_{X|\vartheta}(x|\theta) Q_{y,x} \int_\mathcal{A} l(\theta,a)  P_{A|Y}(da|y).
\end{align}
In practice, the true risk cannot be directly evaluated because the underlying parameter $\theta$ is unknown.
To address this issue, there are two standard frameworks: the Bayesian and minimax frameworks.
In the Bayesian framework, we assume the model parameter is a random variable $\vartheta \in \Theta$ following a prior distribution $\lambda \in \mathcal{P}(\Theta)$. 
The Bayesian risk $R_{\mathrm{B}}$ is defined by
\begin{equation}
    R_{\mathrm{B}}(Q,P_{A|Y};\lambda) := \int_{\Theta} R(\theta,Q,P_{A|Y}) \lambda(d\theta).
\end{equation}
In the minimax framework, we consider the worst-case risk $R_{\mathrm{M}}$, defined by
\begin{equation}
    R_{\mathrm{M}}(Q,P_{A|Y}) := \sup\limits_{\theta \in \Theta} R(\theta,Q,P_{A|Y}).
\end{equation}
For a fixed channel $Q$, we define the optimal Bayesian risk and the minimax risk by minimizing the Bayesian risk and the worst-case risk over all possible decision rules $P_{A|Y}$, respectively:
\begin{align}
    R_{\mathrm{B}}^*(Q;\lambda) &:= \inf\limits_{P_{A|Y} \in \mathcal{P}(\mathcal{Y} \rightarrow \mathcal{A})} R_{\mathrm{B}}(Q,P_{A|Y};\lambda),
    \\ R_{\mathrm{M}}^*(Q) & := \inf\limits_{P_{A|Y} \in \mathcal{P}(\mathcal{Y} \rightarrow \mathcal{A}) } R_{\mathrm{M}}(Q,P_{A|Y}).
\end{align}
For the Bayesian (resp. minimax) framework, the optimal privacy-utility trade-off (PUT) is defined as the minimum optimal Bayesian (resp. minimax) risk over all $\epsilon$-LDP channels:
\begin{align}
    \mathrm{PUT}_{\mathrm{B}}(\mathcal{X},\epsilon;\lambda) &:= \inf\limits_{Q \in \mathcal{Q}_{\mathcal{X},\epsilon}}R_{\mathrm{B}}^*(Q;\lambda),
    \\ \mathrm{PUT}_{\mathrm{M}}(\mathcal{X},\epsilon) &:= \inf\limits_{Q \in \mathcal{Q}_{\mathcal{X},\epsilon}} R^*_{\mathrm{M}}(Q).
\end{align}
We call an $\epsilon$-LDP channel $Q$ that achieves the optimal PUT an optimal $\epsilon$-LDP channel.

In addition, we can consider a statistical decision-making scenario that is invariant under some symmetry, which can be defined formally as follows.

\begin{definition}\label{def:g_invariant_problem}
    {
    Let $G$ be a group.
    An invariant action of $G$ on a statistical decision-making scenario $(\Theta, \mathcal{X}, P_{X|\vartheta}, \mathcal{A}, l)$ consists of $G$-actions on $\Theta, \mathcal{X}$, and $\mathcal{A}$, satisfying the following two conditions:
    \begin{enumerate}
        \item \emph{Model Invariance}:
         $\forall \theta \in \Theta,\; g \in G,\; x \in \mathcal{X}, \quad P_{X|\vartheta}(gx|g\theta) = P_{X|\vartheta}(x|\theta)$.
        \item \emph{Loss Invariance}:
         $\forall \theta \in \Theta,\; g \in G,\; a \in \mathcal{A},\quad l(g \theta, g a) = l(\theta, a)$.
    \end{enumerate}
    When these conditions hold, the scenario is said to be $G$-invariant with respect to such actions.
    Additionally, a prior distribution $\lambda \in \mathcal{P}(\Theta)$ is said to be $G$-invariant if 
    \begin{equation}
        \forall g \in G, \; B \subseteq \Theta, \quad \lambda (g B ) = \lambda(B).
    \end{equation}
    }

\end{definition}

For a given $G$-action on $\mathcal{X}$, the induced $G$-actions on a channel and corresponding invariance is defined as follows.

\begin{definition}\label{def:invariant_ch}
    Let $G$ be a finite group, and its action on a finite input alphabet $\mathcal{X}$ is given.
    For any given finite output alphabet $\mathcal{Y}$ equipped with a $G$-action $\sigma_{(\cdot)}:G \rightarrow (\mathcal{Y} \rightarrow \mathcal{Y})$, we define the induced action on $\mathcal{P}(\mathcal{X}\rightarrow \mathcal{Y})$ by
    \begin{equation}
        (g\circ_\sigma Q)_{y,x} := Q_{\sigma_{g^{-1}}(y),g^{-1}x}.
    \end{equation}
    A channel $Q \in \mathcal{P}(\mathcal{X}\to \mathcal{Y})$ is said to be $G$-invariant with respect to $\sigma$ if $g \circ_\sigma Q = Q$ for all $g \in G$.
    We define $\mathcal{Q}_{\mathcal{X},\epsilon}^G$ as the class of all $\epsilon$-LDP channels from $\mathcal{X}$ to an arbitrary output alphabet $\mathcal{Y}$ that exhibit $G$-invariance for at least one $G$-action $\sigma$.
\end{definition}

\section{Main Results}\label{sec:main}

We provide a unified framework for calculating the optimal PUT and characterizing the optimal LDP channels.
We achieve this by leveraging the functional properties of optimal risks alongside the intrinsic geometry of the set of LDP channels.
The following subsections detail these results.

\subsection{Functional Properties of Optimal Risks}\label{sec:struct}

We first introduce the Blackwell order \cite{blackwellComparisonExperiments1951} which captures the DPI as an order between two channels: post-processing a channel's output cannot increase the amount of information it provides.

\begin{definition}[Blackwell order]
    Let $\mathcal{X}$, $\mathcal{Y}^{(1)}$, and $\mathcal{Y}^{(2)}$ be finite sets.
    For channels $Q^{(1)} \in \mathcal{P}(\mathcal{X}\rightarrow \mathcal{Y}^{(1)})$ and $Q^{(2)} \in \mathcal{P}(\mathcal{X}\rightarrow \mathcal{Y}^{(2)})$, we say $Q^{(1)}$ dominates $Q^{(2)}$, and write $Q^{(2)} \precsim Q^{(1)}$, if there exists a channel $W \in \mathcal{P}(\mathcal{Y}^{(1)} \rightarrow \mathcal{Y}^{(2)})$ such that
    \begin{equation}
        Q^{(2)} = W Q^{(1)}.
    \end{equation}
\end{definition}

The Blackwell order $\precsim$ is a preorder on $\mathcal{P}(\mathcal{X} \to *)$, and this naturally induces an equivalence relation $\sim$, defined by $Q^{(1)} \sim Q^{(2)}$ if and only if $Q^{(1)} \precsim Q^{(2)}$ and $Q^{(2)} \precsim Q^{(1)}$.
For a subclass $\mathcal{Q} \subset \mathcal{P}(\mathcal{X} \to *)$, we denote $\mathcal{Q}/\sim$ by the quotient subclass with respect to this equivalence relation.
Also, with respect to Blackwell order, \emph{most informative channels} can be formally described by maximal elements.

\begin{definition}
    Let $(\mathcal{Q}, \precsim)$ be a preordered class.
    An element $Q \in \mathcal{Q}$ is a maximal element of $\mathcal{Q}$ if
    \begin{equation}
        \forall B \in \mathcal{Q},\; Q \precsim B \Rightarrow B \precsim Q.
    \end{equation}
    The subclass of all maximal elements of $\mathcal{Q}$ is denoted by $\max(\mathcal{Q})$.
\end{definition}

Given the notion of maximality, a natural hypothesis is that the optimal LDP channel $Q$ resides within $\max(\mathcal{Q}_{\mathcal{X},\epsilon})$.
Also, if the underlying statistical decision-making scenario exhibits $G$-invariance, one would expect the optimal LDP channel to preserve this symmetry. Both of these intuitive hypotheses hold true, emerging as consequences of the functional properties of the optimal risks $R^*_{\mathrm{B}}$ and $R^*_{\mathrm{M}}$.

\begin{definition}\label{def:struct}
    A map $R^*: \mathcal{P}(\mathcal{X} \rightarrow *)\rightarrow \mathbb{R}$ is said to satisfy:
    \begin{enumerate}
        \item \textbf{Data-Processing Inequality (DPI)}, if it preserves the Blackwell order in the reverse order:
        \begin{equation}\label{eq:DPI}
            \!\!\!\!\!\! \forall Q, \tilde{Q} \in \mathcal{P}(\mathcal{X} \to *), \; Q \precsim \tilde{Q} \Rightarrow R^*(Q) \geq R^*(\tilde{Q}).
        \end{equation}

        \item \textbf{Direct-Sum Quasi-Convexity (DS-QCVX)}, if for any collection of channels $Q^{(j)} \in \mathcal{P}(\mathcal{X} \rightarrow *)$ and probabilities $p_j \geq 0$ summing to $1$,
        \begin{equation}
            R^*\left( \bigoplus_j p_j Q^{(j)} \right) \leq \max_j R^*(Q^{(j)}).
        \end{equation}

        \item \textbf{Direct-Sum Additivity (DSA)}, if for any collection of channels $Q^{(j)} \in \mathcal{P}(\mathcal{X} \rightarrow *)$ and probabilities $p_j \geq 0$ summing to $1$,
        \begin{equation}
            R^*\left( \bigoplus_j p_j Q^{(j)} \right) = \sum_j p_j R^*(Q^{(j)}).
        \end{equation}

        \item \textbf{Concavity (CCV)}, if for any finite set $\mathcal{Y}$, channels $Q^{(j)} \in \mathcal{P}(\mathcal{X} \rightarrow \mathcal{Y})$, and probabilities $p_j \geq 0$ summing to $1$,
        \begin{equation}
            R^*\left( \sum_j p_j Q^{(j)} \right) \geq \sum_j p_j R^*(Q^{(j)}).
        \end{equation}
    \end{enumerate}
\end{definition}

\begin{remark}
    DPI and DSA imply CCV, and DSA trivially implies DS-QCVX.
    However, DPI and DS-QCVX do not imply CCV in general.
\end{remark}

\begin{definition}\label{def:invariant_utility}
    Let $G$ be a finite group whose action on a finite set $\mathcal{X}$ is given.
    A map $R^*: \mathcal{P}(\mathcal{X} \rightarrow *) \rightarrow \mathbb{R}$ is $G$-invariant if for any finite output alphabet $\mathcal{Y}$ and any $G$-action $\sigma$ on $\mathcal{Y}$,
    \begin{equation}
        \forall g \in G,\; Q \in \mathcal{P}(\mathcal{X} \rightarrow \mathcal{Y}), \quad R^*(g \circ_\sigma Q) = R^*(Q).
    \end{equation}
\end{definition}

\begin{proposition}\label{prop:BM_structure}
    For a privacy-preserving statistical decision-making scenario introduced in Section~\ref{sec:stat_dec}, the following holds:
    \begin{enumerate}
        \item For any given prior distribution $\lambda \in \mathcal{P}(\Theta)$, the optimal Bayesian risk $R^*_{\mathrm{B}}$ satisfies the DPI and DSA.

        \item The minimax risk $R^*_{\mathrm{M}}$ satisfies the DPI and DS-QCVX.

        \item If a statistical decision-making scenario is $G$-invariant, then $R^*_{\mathrm{B}}$ with $G$-invariant prior $\lambda$ and $R^*_{\mathrm{M}}$ are $G$-invariant.
    \end{enumerate}
\end{proposition}

The formal proof of Proposition~\ref{prop:BM_structure} is provided in Appendix~\ref{sec:pf_BM_structure}.
Briefly outlined, the proof is based on the following idea.
Regarding the DPI, post-processing can be subsumed as a part of a decision rule.
Concerning the DSA and DS-QCVX, we can interpret the direct sum of channels $\bigoplus_j p_j Q^{(j)}$ as a \emph{flagged mixture} of the $Q^{(j)}$'s.
That is, it first samples a random variable $J$, taking the value $j$ with probability $p_j$.
Given $J=j$, it operates as $Q^{(j)}$ while providing the flag $J=j$ to the decision maker.
Then, based on the knowledge $J=j$, the decision maker performs the optimal decision for $Q^{(j)}$.
Finally, regarding $G$-invariance, for a translated channel $g\circ_\sigma Q$, the decision maker can apply the optimal decision rule for $Q$ combined with the analogous translation.
The resulting risk coincides with that of the original channel $Q$ under its optimal decision rule, which is also identical to the risk at the translated model parameter due to the model and loss invariance of the scenario.

Then, in the following theorem, we show that the optimization domain for computing the optimal PUT can be reduced based on the functional properties of the risks presented in Proposition~\ref{prop:BM_structure}.

\begin{theorem}\label{thm:opt_PUT_red}
    Suppose a map $R^*: \mathcal{P}(\mathcal{X} \rightarrow *)\rightarrow \mathbb{R}$ satisfies the DPI.
    Then,
    \begin{equation}\label{eq:suff_max}
        \inf\limits_{Q \in \mathcal{Q}_{\mathcal{X},\epsilon}}R^* (Q) = \inf\limits_{Q \in \max(\mathcal{Q}_{\mathcal{X},\epsilon})}R^* (Q).
    \end{equation}
    Moreover, if $R^*$ is also $G$-invariant and satisfies DS-QCVX, then 
    \begin{equation}\label{eq:suff_G_inv_max}
        \inf\limits_{Q \in \mathcal{Q}_{\mathcal{X},\epsilon}}R^* (Q) = \inf\limits_{Q \in \max^G(\mathcal{Q}_{\mathcal{X},\epsilon})}R^* (Q),
    \end{equation}
    where $\max^G(\mathcal{Q}_{\mathcal{X},\epsilon}) := \max(\mathcal{Q}_{\mathcal{X},\epsilon}) \cap \mathcal{Q}_{\mathcal{X},\epsilon}^G$.
\end{theorem}

One might see the above theorem as a direct consequence of the definitions, but there are two non-trivial points.
First, for a general preordered class, there might exist an element that is not dominated by a maximal element.
Second, even if we can restrict the optimization domain to each of the two subdomains, $A$ and $B$, this does not guarantee a reduction to $A \cap B$.
These two issues are resolved by the geometric characterization of the maximal LDP channels, which is developed in the following subsection.
The proof is provided in Appendix~\ref{sec:pf_opt_PUT_red}.

\begin{remark}\label{rmk:info_meas}
    Although we developed our framework within the context of privacy-preserving statistical decision-making—where $R^*$ represents the optimal Bayesian or minimax risk—it is also applicable to maximizing utility functions based on information-theoretic measures. 
    For example, $f$-divergences between induced output distributions \cite{kairouzExtremalMechanismsLocal2016b}, the mutual information between the input and output of an LDP channel $Q$ under a uniformly distributed input $X$ \cite{wangLocalDifferentialPrivate2019}, and the Fisher information of the output $Y$ with respect to a parameter $\theta$ \cite{barnesFisherInformationLocal2020}, all satisfy the DPI (with the inequality reversed relative to \eqref{eq:DPI}) and DSA.
\end{remark}

\subsection{Characterizing Maximal LDP Channels}\label{sec:max_LDP}
While Theorem~\ref{thm:opt_PUT_red} reduces the domain of optimization in calculating the optimal PUT, it remains unclear whether there is a computationally tractable way to calculate it.
One of the main challenges is that $\max(\mathcal{Q}_{\mathcal{X},\epsilon})$ contains a matrix with arbitrarily large dimension in general.
As a remedy, we first establish a one-to-one correspondence between $\max(\mathcal{Q}_{\mathcal{X},\epsilon})/\sim$ and a specific finite-dimensional polytope $\mathcal{S}_{\mathcal{X},\epsilon}$.
This characterization is closely related to the concept of extremal mechanisms \cite{kairouzExtremalMechanismsLocal2016b}, which were previously shown to form a cofinal set\footnote{A subclass $\mathcal{B} \subseteq \mathcal{A}$ is called a cofinal class of a preordered set $\mathcal{A}$ if $\forall A \in \mathcal{A},\; \exists B \in \mathcal{B},\; A \precsim B$. In general, any cofinal set contains $\max(\mathcal{A})$, but the converse does not hold.} of $\max(\mathcal{Q}_{\mathcal{X},\epsilon})/\sim$ \cite{namQuantumAdvantagePrivate2025}.

We note that the results in \cite[Thms.~4.2, 4.3]{chaiCriterionPrivacyProtection2018a} proved that a column stochastic matrix $Q$ is in $\max(\mathcal{Q}_{\mathcal{X},\epsilon})$ if and only if every rows of $Q$ is binary-valued, and 
\begin{equation}
    \forall y \in \mathcal{Y}, \quad \frac{\max_x Q_{y,x}}{\min_x Q_{y,x}}  = e^\epsilon.
\end{equation}
Similar result was also provided in \cite{fernandesUniversalOptimalityRobust2023b}.
In this work, we reinterpret this result through an alternative convex geometric derivation, treating the likelihood vectors of an LDP channel as elements within a cone.

\begin{definition}[Extremal LDP channels\footnote{The original work \cite{kairouzExtremalMechanismsLocal2016b} considered $\mathcal{Y} = 2^{\mathcal{X}}$. Theorem~\ref{thm:extremal_LDP_is_max} shows that an extremal LDP channel containing an all-ones row (up to a multiplicative factor) is not maximal.}]
    Let $\mathcal{B}(\mathcal{X}) = 2^{\mathcal{X}} \setminus \{\emptyset, \mathcal{X}\}$.
    The staircase matrix $S_{\mathcal{X},\epsilon} \in \{1,e^\epsilon\}^{\mathcal{B}(\mathcal{X}) \times \mathcal{X}}$ is defined by 
    \begin{equation}\label{eq:staircase_mat}
        (S_{{\mathcal{X}},\epsilon})_{y,x} := \begin{cases}
            e^\epsilon & \text{if } x \in y 
            \\ 1 & \text{otherwise}
        \end{cases}.
    \end{equation}
    Let the set of weight vectors, which we refer to as the maximal $\epsilon$-LDP polytope, be defined as
    \begin{equation}\label{eq:stot_const_staircase}
        \mathcal{S}_{\mathcal{X},\epsilon} := \{c \in \mathbb{R}_+^{\mathcal{B}(\mathcal{X})} \mid c^\top S_{\mathcal{X},\epsilon} = \mathbf{1}_\mathcal{X}^\top \}.
    \end{equation}
    Then, we define a map $Q_{\mathrm{ex}}^{(\cdot)}$ defined on $\mathcal{S}_{\mathcal{X},\epsilon}$ by
    \begin{equation}\label{eq:extremal_LDP_ch}
        Q_{\mathrm{ex}}^{(c)} := \bigoplus_{y \in \mathcal{B}(\mathcal{X})} c_y(S_{\mathcal{X},\epsilon})_{y,*}.
    \end{equation}
    For $c \in \mathcal{S}_{\mathcal{X},\epsilon}$, we say $Q_{\mathrm{ex}}^{(c)}$ is an extremal $\epsilon$-LDP channel with weight vector $c$.
\end{definition}

\begin{theorem}[cf.~\cite{chaiCriterionPrivacyProtection2018a, fernandesUniversalOptimalityRobust2023b}] \label{thm:extremal_LDP_is_max}
    The map $Q_{\mathrm{ex}}^{(\cdot)}$ defined in \eqref{eq:extremal_LDP_ch} establishes a one-to-one correspondence between $\mathcal{S}_{\mathcal{X},\epsilon}$ and $\max(\mathcal{Q}_{\mathcal{X},\epsilon})/\sim$.
\end{theorem}

The detailed proof of the above theorem is provided at the end of  Section~\ref{sec:pf_extremal}.
As a brief sketch, we first show that an $\epsilon$-LDP channel $Q$ is a maximal LDP channel if and only if every non-zero row of it corresponds to an extreme direction of a certain polyhedral cone defining the LDP constraint (Proposition~\ref{prop:max_LDP_geo}).
Then, we show that such an extreme direction corresponds to a non-empty proper subset of $\mathcal{X}$ (Lemma~\ref{lem:ext_ray_LDP}).

As a direct consequence of Theorem~\ref{thm:extremal_LDP_is_max}, in minimizing $R^*(Q)$ satisfying the DPI over $Q \in \mathcal{Q}_{\mathcal{X},\epsilon}$, it is sufficient to consider a finite-dimensional polytope $\mathcal{S}_{\mathcal{X},\epsilon}$, i.e.,
\begin{equation}\label{eq:red_poly}
    \inf\limits_{Q \in \mathcal{Q}_{\mathcal{X},\epsilon}} R^*(Q) = \min\limits_{c \in \mathcal{S}_{\mathcal{X},\epsilon}} R^*\left(Q_{\mathrm{ex}}^{(c)}\right).
\end{equation}

\begin{remark}\label{rmk:LP_DPI}
    If $R^*$ is also concave, then the above optimization problem can be computed by vertex enumeration.
    Also, if $R^*$ admits an extension to $\mathbb{R}_+^{* \times \mathcal{X}}$ that also satisfies DSA over $\mathbb{R}_+^{* \times \mathcal{X}}$, such as the optimal Bayesian risk $R_{\mathrm{B}}^*$, then the above optimization problem can be computed by linear programming (cf. \cite[Thm.~4]{kairouzExtremalMechanismsLocal2016b}).
\end{remark}

Theorem~\ref{thm:extremal_LDP_is_max} also directly recovers the previous result \cite[Thm.~18]{kairouzExtremalMechanismsLocal2016b}\footnote{The result in \cite[Thm.~18]{kairouzExtremalMechanismsLocal2016b} considers more general privacy constraint, $(\epsilon,\delta)$-LDP, where $(\epsilon,0)$-LDP is the same as $\epsilon$-LDP.}: whenever the input data is binary ($|\mathcal{X}|=2$), then binary \emph{randomized response} \cite{warnerRandomizedResponseSurvey1965} is the unique maximal LDP channel up to equivalence.
We formally state this fact in the following corollary, whose proof is omitted.
\begin{corollary}[{\cite[Prop.~4.1]{chaiCriterionPrivacyProtection2018a}, \cite[Cor.~18]{fernandesUniversalOptimalityRobust2023b}}]
    If $|\mathcal{X}| = 2$, then $Q \in \max(\mathcal{Q}_{\mathcal{X},\epsilon})$ if and only if 
    \begin{equation}
        Q \sim \frac{1}{e^\epsilon + 1}\begin{pmatrix}
            e^\epsilon & 1 \\ 1 & e^\epsilon
        \end{pmatrix}.
    \end{equation}
\end{corollary}

\subsubsection{Geometric characterization of maximal LDP channels}

In this subsection, we provide a geometric characterization of $\max(\mathcal{Q}_{\mathcal{X},\epsilon})$.
Some background on convex geometry used throughout this work are introduced in Section~\ref{sec:conv_geo}.
For LDP channels $Q \in \mathcal{Q}_{\mathcal{X},\epsilon}$, it is convenient to consider their rows (likelihood vectors) $\{Q_{y,*}\}_{y \in \mathcal{Y}}$ since the LDP constraints act independently on each row.

\begin{definition}
    The $\epsilon$-LDP cone $\mathcal{K}_{\mathcal{X},\epsilon}$ is defined as 
    \begin{equation}
        \mathcal{K}_{\mathcal{X},\epsilon} := \left\{ v \in \mathbb{R}_+^\mathcal{X} \;\middle|\; \forall x,x' \in \mathcal{X}, \; e^\epsilon v_x - v_{x'} \geq 0 \right\}.
    \end{equation}
\end{definition}

With the above definition, it can be easily seen that a column stochastic matrix $Q$ is an $\epsilon$-LDP channel if and only if all of its rows lie inside the cone $\mathcal{K}_{\mathcal{X},\epsilon}$.
We formally state this fact and some properties of $\mathcal{K}_{\mathcal{X},\epsilon}$ in the following lemma, whose proof is omitted.
\begin{lemma}\label{lem:LDP_cone}
    Let $Q \in \mathcal{P}(\mathcal{X} \rightarrow \mathcal{Y})$.
    We have that 
    \begin{equation}
        Q \in \mathcal{Q}_{\mathcal{X},\epsilon} \Leftrightarrow \forall y \in \mathcal{Y},\; (Q_{y,*})^\top \in \mathcal{K}_{\mathcal{X},\epsilon}.
    \end{equation}
    Furthermore, $\mathcal{K}_{\mathcal{X},\epsilon}$ is a nontrivial, pointed, polyhedral cone defined by a constraint matrix $T_{\mathcal{X},\epsilon} \oplus I_{\mathcal{X}\times \mathcal{X}}$, where $T_{\mathcal{X},\epsilon} \in \mathbb{R}^{\mathcal{I}_{\mathcal{X}} \times \mathcal{X}}$,
    \begin{equation}\label{eq:LDP_const_mat}
        (T_{\mathcal{X},\epsilon})_{(z,z'), x} := 
        \begin{cases} 
            e^\epsilon & \text{if } x = z \\
            -1 & \text{if } x = z' \\
            0 & \text{otherwise}
        \end{cases},
    \end{equation}
    and $\mathcal{I}_{\mathcal{X}} := \{(z, z') \in \mathcal{X} \times \mathcal{X} \mid z \neq z'\}$.
\end{lemma}

Now, we first show that maximality of an LDP channel can be equivalently described by the extremality of each non-zero row.

\begin{proposition}\label{prop:max_LDP_geo}
    Let $Q \in \mathcal{Q}_{\mathcal{X},\epsilon}$ be a channel with output alphabet $\mathcal{Y}$.
    Then, $Q \in \max(\mathcal{Q}_{\mathcal{X},\epsilon})$ if and only if every non-zero row of $Q$ is an extreme direction of $\mathcal{K}_{\mathcal{X},\epsilon}$.
    In addition, every $Q \in \mathcal{Q}_{\mathcal{X},\epsilon}$ is dominated by a maximal LDP channel in $\max(\mathcal{Q}_{\mathcal{X},\epsilon})$.
\end{proposition}

Before proving the proposition, we state the following equivalent relations that are used in the proof.
We omit the proof of the following lemma since it is easy to verify.

\begin{lemma}\label{lem:equiv_ch}
    The following equivalences hold:
    \begin{enumerate}
        \item Relabeling: Let $Q \in \mathcal{P}(\mathcal{X} \rightarrow \mathcal{Y}_1)$, $\phi:\mathcal{Y}_1 \rightarrow \mathcal{Y}_2$ be a bijection, and $Q' \in \mathcal{P}(\mathcal{X} \rightarrow \mathcal{Y}_2)$ be the induced channel defined by $Q'_{y,x} = Q_{\phi^{-1}(y),x}$.
        Then, $Q \sim Q'$.
        
        \item Zero-padding: Let $Q \in \mathcal{P}(\mathcal{X} \rightarrow \mathcal{Y}_1)$ and let $\mathcal{Y}_2$ be a finite set.
        Then, $Q \oplus \mathbf{0}_{\mathcal{Y}_2,\mathcal{X}} \sim Q$.

        \item Row-splitting: Let $Q \in \mathcal{P}(\mathcal{X} \rightarrow \mathcal{Y})$.
        For any given $y_0 \in \mathcal{Y}$ and $\lambda \in [0,1]$, we have
        \begin{equation}
            Q \sim (\oplus_{y \in \mathcal{Y}\setminus\{y_0\}} Q_{y,*}) \oplus (\lambda Q_{y_0}) \oplus ((1-\lambda) Q_{y_0}).
        \end{equation}
    \end{enumerate}
\end{lemma}

Now, we prove Proposition~\ref{prop:max_LDP_geo}.

\begin{IEEEproof}[Proof of Proposition~\ref{prop:max_LDP_geo}]
    We first prove the if and only if part.
    
    ($\Rightarrow$)
    From Proposition~\ref{prop:carath_cone} in Section~\ref{sec:conv_geo}, the cone $\mathcal{K}_{\mathcal{X},\epsilon}$ has a finite number of extreme rays, denoted as $\mathcal{R}(v_1), \dots, \mathcal{R}(v_l)$.
    Also, any element in $\mathcal{K}_{\mathcal{X},\epsilon}$ is a conic combination of extreme rays.
    That is, for every $y$, there exist non-negative constants $c_{y,1}, \dots, c_{y,l}$ such that
    \begin{equation}\label{eq:pf_prop_b3_1}
        Q_{y,*} = \sum_{j \in [l]} c_{y,j} (v_j)^\top.    
    \end{equation}
    Now, we construct a channel $\tilde{Q} \in \mathcal{P}(\mathcal{X} \rightarrow \tilde{\mathcal{Y}})$, $\tilde{\mathcal{Y}} = \mathcal{Y} \times [l]$, by splitting each summand of each row of $Q$ into distinct rows, i.e.,
    \begin{equation}\label{eq:Q_tilde}
        \tilde{Q} = \bigoplus_{(y,j) \in \mathcal{Y} \times [l]} c_{y,j} (v_j)^\top.
    \end{equation}
    Then, it is easy to check that $\tilde{Q} \in \mathcal{Q}_{\mathcal{X},\epsilon}$, and $Q = V\tilde{Q}$ holds for the column-stochastic matrix $V \in \mathcal{P}(\tilde{\mathcal{Y}} \rightarrow \mathcal{Y})$ given by
    \begin{equation}\label{eq:pf_prop_b3_2}
        V_{y',(y,j)}=\mathbbm{1}(y=y').
    \end{equation}
    Because $\tilde{Q} \succsim Q$ and $Q \in \max(Q_{\mathcal{X},\epsilon})$, $\tilde{Q} \precsim Q$ should hold, i.e., there exists a column-stochastic matrix $W$ such that $\tilde{Q} = W Q$.

    Now, for contradiction, assume that there exists a non-zero row of $Q$ which is not an extreme direction, say $Q_{y^*,*}$.
    Then, $Q_{y^*,*}$ is a nontrivial conic combination of distinct extreme directions.
    Because $\tilde{Q} = WQ$, $W_{\tilde{y},y^*} = 0$ should hold for all $\tilde{y}$ (otherwise, a row of $\tilde{Q}$ is also a nontrivial conic combination of distinct extreme directions).
    This implies that a column $W_{*,y^*}$ is all-zero column, which contradicts the column-stochasticity of $W$.
    By contradiction, we can conclude that every non-zero row of $Q$ is an extreme direction.

    ($\Leftarrow$)
    Let $\{\mathcal{R}(v_1), \dots, \mathcal{R}(v_l)\}$ be the set of all distinct extreme rays of $\mathcal{K}_{\mathcal{X},\epsilon}$.
    Let $Q \in \mathcal{Q}_{\mathcal{X},\epsilon}$ be a channel with output alphabet $\mathcal{Y}$ whose non-zero rows are extreme directions.
    We define the \emph{gathered} channel $P \in \mathcal{P}(\mathcal{X} \rightarrow [l])$ by pooling the rows of $Q$ along each ray:
    \begin{equation}
        P_{j,*} := \sum_{y: Q_{y,*} \in \mathcal{R}(v_j)} Q_{y,*}.    
    \end{equation}
    Because $Q$ can be perfectly reconstructed from $P$ by proportionally dividing each row $P_{j,*}$ and appending any all-zero rows, repeated applications of the row splitting and zero-padding operations in Lemma \ref{lem:equiv_ch} imply $Q \sim P$.
    
    Now, suppose there exists $\tilde{Q} \in \mathcal{Q}_{\mathcal{X},\epsilon}$ with output alphabet $\tilde{\mathcal{Y}}$ such that $Q \precsim \tilde{Q}$, meaning $Q = W\tilde{Q}$ for some column-stochastic matrix $W \in \mathcal{P}(\tilde{\mathcal{Y}} \rightarrow \mathcal{Y})$.
    Because the rows of $Q$ are extreme directions, the relation $Q = W\tilde{Q}$ implies that every row of $\tilde{Q}$ is also an extreme direction, and $W_{y,\tilde{y}}>0$ only if both $Q_{y,*}$ and $\tilde{Q}_{\tilde{y},*}$ are in the same extreme ray.
    Furthermore, because $W$ is column-stochastic, the total probability mass along each extreme ray is conserved between $Q$ and $\tilde{Q}$, i.e.,
    \begin{align}
        P_{j,*} &= \sum_{y: Q_{y,*} \in \mathcal{R}(v_j)} Q_{y,*}= \sum_{y: Q_{y,*} \in \mathcal{R}(v_j)} \sum_{\tilde{y} \in \tilde{\mathcal{Y}}} W_{y,\tilde{y}} \tilde{Q}_{\tilde{y},*}
        \\&= \sum_{\tilde{y}: \tilde{Q}_{\tilde{y},*} \in \mathcal{R}(v_j)} \left( \sum_{y \in \mathcal{Y}} W_{y,\tilde{y}} \right) \tilde{Q}_{\tilde{y},*}
        \\& = \sum_{\tilde{y}: \tilde{Q}_{\tilde{y},*} \in \mathcal{R}(v_j)} \tilde{Q}_{\tilde{y},*}.
    \end{align}
    Consequently, gathering the rows of $\tilde{Q}$ along each ray yields the exact same gathered channel $P$.
    By Lemma \ref{lem:equiv_ch}, this guarantees $\tilde{Q} \sim P$.
    Since $Q \sim P \sim \tilde{Q}$, we can conclude that $Q \in \max(\mathcal{Q}_{\mathcal{X},\epsilon})$.

    Regarding dominance, the first part of the proof shows that $\tilde{Q}$ in \eqref{eq:Q_tilde} dominates $Q \in \mathcal{Q}_{\mathcal{X},\epsilon}$, where $\tilde{Q} \in \max(\mathcal{Q}_{\mathcal{X},\epsilon})$ because all of its rows are extreme directions. 
\end{IEEEproof}

\begin{remark}\label{rmk:general_const}
    Not only for $\epsilon$-LDP, Proposition~\ref{prop:max_LDP_geo} can also be extended to more general constraints such as metric DP \cite{chatzikokolakisBroadeningScopeDifferential2013, fernandesUniversalOptimalityRobust2023b}.\footnote{Let $d$ be a metric on $\mathcal{X}$. A channel $Q \in \mathcal{P}(\mathcal{X} \rightarrow \mathcal{Y})$ is said to satisfy metric DP with respect to $d$ if $e^{d(x,x')} Q_{y,x} - Q_{y,x'} \geq 0$ holds for all $y,x,x'$.}
    An analogous result holds if a constraint on a channel $Q$ is defined by requiring all of its rows to lie within a pointed nontrivial polyhedral cone.
\end{remark}

\subsubsection{Subset characterization of extreme directions}\label{sec:pf_extremal}
To give an explicit characterization of $\mathrm{max}(\mathcal{Q}_{\mathcal{X},\epsilon})$ beyond Proposition~\ref{prop:max_LDP_geo}, we show that there is a one-to-one correspondence between the set of all extreme rays of $\mathcal{K}_{\mathcal{X},\epsilon}$ and the set of all non-empty proper subsets of $\mathcal{X}$.

\begin{lemma}\label{lem:ext_ray_LDP}
    A vector $v \in \mathcal{K}_{\mathcal{X},\epsilon}$ is an extreme direction of $\mathcal{K}_{\mathcal{X},\epsilon}$ if and only if there exists a non-empty proper subset $\mathcal{Z} \subsetneq \mathcal{X}$ and a constant $c>0$ such that
    \begin{equation}\label{eq:ex_ray_LDP}
        v_x = \begin{cases}
            ce^\epsilon & \text{if } x \in \mathcal{Z} \\
            c & \text{if } x \notin \mathcal{Z}.
        \end{cases}
    \end{equation}
\end{lemma}

\begin{IEEEproof}
    Based on Proposition~\ref{prop:ex_ray_rank} in Section~\ref{sec:conv_geo}, we prove the lemma by showing the equivalence with $\mathrm{ker}\left(T_{\mathcal{X},\epsilon}^{(v)}\right) = \mathrm{span}\left( \{v\} \right)$.
    
    ($\Leftarrow$) 
    By definition, $\mathrm{ker}\left(T_{\mathcal{X},\epsilon}^{(v)}\right) \supseteq \mathrm{span}\left( \{v\} \right)$.
    We will show $\mathrm{ker}\left(T_{\mathcal{X},\epsilon}^{(v)}\right) \subseteq \mathrm{span}\left( \{v\} \right)$.
    Suppose there exists $u \in \mathrm{ker}\left(T_{\mathcal{X},\epsilon}^{(v)}\right)$.
    By \eqref{eq:ex_ray_LDP}, the active indices for $v$ are exactly the pairs $(z, z') \in \mathcal{X} \times \mathcal{X}$ where $z \notin \mathcal{Z}$ and $z' \in \mathcal{Z}$, i.e.,
    \begin{equation}
        T_{\mathcal{X},\epsilon}^{(v)} = \bigoplus_{z \notin \mathcal{Z}, z' \in \mathcal{Z}} T_{(z,z'),*}.
    \end{equation}
    Because $u \in \mathrm{ker}\left(T_{\mathcal{X},\epsilon}^{(v)}\right)$, we have
    \begin{equation}
        \forall z \notin \mathcal{Z}, z' \in \mathcal{Z}, \quad e^\epsilon u_z - u_{z'} = 0.
    \end{equation}
    Fix an element $z_0 \notin \mathcal{Z}$.
    Then for any $z' \in \mathcal{Z}$, we have $u_{z'} = e^\epsilon u_{z_0}$.
    This implies that $u$ is constant on $\mathcal{Z}$.
    Similarly, fix an element $z'_0 \in \mathcal{Z}$.
    Then for any $z \notin \mathcal{Z}$, we have $e^\epsilon u_z = u_{z'_0}$, meaning $u_z = e^{-\epsilon} u_{z'_0}$.
    Thus, $u$ is also constant on $\mathcal{X} \setminus \mathcal{Z}$.
    Let $u_z = k$ for $z \notin \mathcal{Z}$. It follows that $u_{z'} = k e^\epsilon$ for $z' \in \mathcal{Z}$.
    Therefore, $u = (k/c) v$, meaning $u$ is proportional to $v$. 
    Hence, we conclude that $\mathrm{ker}\left(T_{\mathcal{X},\epsilon}^{(v)}\right) = \mathrm{span}\left( \{v\} \right)$.
    
    ($\Rightarrow$)
    Let $a = \min_x v_x$ and $b = \max_x v_x$, and define 
    \begin{equation}
        \mathcal{W} = \{x \in \mathcal{X} : v_x = a\}, \quad \mathcal{Z} = \{x \in \mathcal{X} : v_x = b\}.
    \end{equation}
    It can be checked that $v_x > 0$ for all $x \in \mathcal{X}$, meaning both $\mathcal{W}$ and $\mathcal{Z}$ are non-empty.
    To prove the desired result, we will show $b = e^\epsilon a$ and $\mathcal{W} = \mathcal{X} \setminus \mathcal{Z}$.
    
    Since $v \in \mathcal{K}_{\mathcal{X},\epsilon}$, we must have $b \leq e^\epsilon a$.
    If $b < e^\epsilon a$, then $e^\epsilon v_z - v_{z'} \geq e^\epsilon a - b > 0$ for all $z \neq z' \in \mathcal{X}$.
    Thus, there is no active constraint for $v$, meaning $\mathrm{ker}\left(T_{\mathcal{X},\epsilon}^{(v)}\right) = \mathbb{R}^\mathcal{X}$.
    This contradicts $\mathrm{ker}\left(T_{\mathcal{X},\epsilon}^{(v)}\right) = \mathrm{span}\left( \{v\} \right)$.
    Hence, we must have $b = e^\epsilon a$.
    
    Now, suppose there exists an element $x_0 \notin \mathcal{W} \cup \mathcal{Z}$, meaning $v_{x_0} \in (a, b)$.
    Because $b = e^\epsilon a$, we have 
    \begin{equation}
        (T_{\mathcal{X},\epsilon})_{(z,z'),*} v = 0 \iff (z,z') \in \mathcal{W} \times \mathcal{Z}.
    \end{equation}
    Since $x_0 \notin \mathcal{W} \cup \mathcal{Z}$, the $x_0$-th column of $T_{\mathcal{X},\epsilon}^{(v)}$ is the zero column.
    This implies the standard basis vector $e_{x_0}$ is in $\mathrm{ker}\left(T_{\mathcal{X},\epsilon}^{(v)}\right)$.
    Since $v_x > 0$ for all $x\in\mathcal{X}$, $v$ and $e_{x_0}$ are linearly independent.
    Therefore, $\mathrm{ker}\left(T_{\mathcal{X},\epsilon}^{(v)}\right)$ has dimension at least 2.
    This again contradicts $\mathrm{ker}\left(T_{\mathcal{X},\epsilon}^{(v)}\right) = \mathrm{span}\left( \{v\} \right)$.
    Thus, $\mathcal{X} = \mathcal{W} \cup \mathcal{Z}$, which dictates that $\mathcal{W} = \mathcal{X} \setminus \mathcal{Z}$, completing the proof.
\end{IEEEproof}

Finally, based on Proposition~\ref{prop:max_LDP_geo} and Lemma~\ref{lem:ext_ray_LDP}, we prove the one-to-one correspondence in Theorem~\ref{thm:extremal_LDP_is_max}.
\begin{IEEEproof}[Proof of Theorem~\ref{thm:extremal_LDP_is_max}]
    We first prove surjectivity.
    Let $Q \in \max(\mathcal{Q}_{\mathcal{X},\epsilon})$ with output alphabet $\mathcal{Z}$.
    By Proposition~\ref{prop:max_LDP_geo} and Lemma~\ref{lem:ext_ray_LDP}, every non-zero row of $Q$ is a constant multiple of a row of $S_{\mathcal{X},\epsilon}$, say $Q_{z,*} = q_z(S_{\mathcal{X},\epsilon})_{y_z,*}$ for some scalar $q_z > 0$ and index $y_z \in \mathcal{B}(\mathcal{X})$.
    For each subset $y \in \mathcal{B}(\mathcal{X})$, we gather all rows of $Q$ corresponding to the same extreme ray by defining a combined weight
    \begin{equation}
        c_y := \sum_{z: y_z = y} q_z,
    \end{equation}
    where $c_y = 0$ if no such $z$ exists.
    This constructs a weight vector $c \in \mathbb{R}_+^{\mathcal{B}(\mathcal{X})}$.
    Because $Q$ is column-stochastic, we have
    \begin{align}
        \sum_{z \in \mathcal{Z}} Q_{z,x} &= \sum_{z \in \mathcal{Z}} q_z (S_{\mathcal{X},\epsilon})_{y_z, x}
        \\&= \sum_{y \in \mathcal{B}(\mathcal{X})} \left( \sum_{z: y_z = y} q_z \right) (S_{\mathcal{X},\epsilon})_{y, x}
        \\& = \sum_{y \in \mathcal{B}(\mathcal{X})} c_y (S_{\mathcal{X},\epsilon})_{y, x} = 1,
    \end{align}
    for all $x \in \mathcal{X}$.
    Furthermore, the process of gathering proportional rows and permuting them does not change the equivalence class of the channel by Lemma~\ref{lem:equiv_ch}.
    Therefore, $Q \sim Q_{\mathrm{ex}}^{(c)}$.

    For injectivity, suppose $Q_{\mathrm{ex}}^{(c)} \sim Q_{\mathrm{ex}}^{(c')}$, meaning there exist column-stochastic matrices $W$ and $V$ such that $Q_{\mathrm{ex}}^{(c)} = W Q_{\mathrm{ex}}^{(c')}$ and $Q_{\mathrm{ex}}^{(c')} = V Q_{\mathrm{ex}}^{(c)}$.
    By construction, we have $(Q_{\mathrm{ex}}^{(c)})_{y,*} = c_y (S_{\mathcal{X},\epsilon})_{y,*}$ and $(Q_{\mathrm{ex}}^{(c')})_{y,*} = c'_y (S_{\mathcal{X},\epsilon})_{y,*}$.
    Because $Q_{\mathrm{ex}}^{(c)} = W Q_{\mathrm{ex}}^{(c')}$ and every row of $S_{\mathcal{X},\epsilon}$ are distinct extreme directions of $\mathcal{K}_{\mathcal{X},\epsilon}$, we get $W_{y,\tilde{y}} c'_{\tilde{y}} = 0$ for all $\tilde{y} \neq y$, and hence $c_y = W_{y,y} c'_y$ for all $y \in \mathcal{B}(\mathcal{X})$.
    By identical logic using the reverse relation $Q_{\mathrm{ex}}^{(c')} = V Q_{\mathrm{ex}}^{(c)}$, we also have $c'_y = V_{y,y} c_y$ for all $y \in \mathcal{B}(\mathcal{X})$. 
    Therefore, we have $c_y = W_{y,y} V_{y,y} c_y$, thereby $W_{y,y} = V_{y,y} = 1$ must hold whenever $c_y > 0$.
    Returning to $c_y = W_{y,y} c'_y$, this immediately gives $c_y = c'_y$ for all $y$ with $c_y > 0$. 
    For any index where $c_y = 0$, the relation $c'_y = V_{y,y} c_y$ ensures $c'_y = 0$ as well.
    Thus, $c = c'$ holds, which completes the proof.
\end{IEEEproof}

\subsection{Characterizing Group Invariant Maximal LDP Channels}\label{sec:G_inv}

While Theorem~\ref{thm:extremal_LDP_is_max} implies that the optimal PUT can be computed by optimizing over finite dimensional polytope $\mathcal{S}_{\mathcal{X},\epsilon}$ as in \eqref{eq:red_poly}, the dimension of $\mathcal{S}_{\mathcal{X},\epsilon}$ increases exponentially with $|\mathcal{X}|$, which in turn makes the computation challenging.
If the objective inherits a symmetry, we can further reduce the domain to a lower dimensional polytope.
To show this fact, we characterize $\max^G(\mathcal{Q}_{\mathcal{X},\epsilon})/\sim$ as a polytope with a dimension lower than that of $\mathcal{S}_{\mathcal{X},\epsilon}$.

From the one-to-one correspondence in Theorem~\ref{thm:extremal_LDP_is_max}, a $G$-action on $\mathcal{Q}_{\mathcal{X},\epsilon}$ induces a $G$-action on $\mathcal{S}_{\mathcal{X},\epsilon}$, and vice versa.
Thereby, the notion of $G$-invariance of a channel $Q \in \mathcal{Q}_{\mathcal{X},\epsilon}$ is translated into the notion of $G$-invariance of a weight vector $c \in \mathcal{S}_{\mathcal{X},\epsilon}$.
We first introduce an induced $G$-action on the maximal $\epsilon$-LDP polytope $\mathcal{S}_{\mathcal{X},\epsilon}$.
Because $\mathcal{S}_{\mathcal{X},\epsilon}$ is defined based on $\mathcal{B}(\mathcal{X})$—the collection of non-empty proper subsets of $\mathcal{X}$—there is a natural induced $G$-action on $\mathcal{B}(\mathcal{X})$ defined by $gy := \{gx \mid x \in y\}$.
This action on $\mathcal{B}(\mathcal{X})$, in turn, induces a $G$-action on $\mathcal{S}_{\mathcal{X},\epsilon}$ defined by $(g c)_y := c_{g^{-1}y}$.

\begin{definition}
    A weight vector $c \in \mathcal{S}_{\mathcal{X},\epsilon}$ is $G$-invariant if $gc = c$ for all $g \in G$.
    Equivalently, $c_y = c_{y'}$ holds for every orbit $\mathcal{O} \in \mathcal{B}(\mathcal{X})/G$ and $y,y' \in \mathcal{O}$.
\end{definition}

\begin{definition}\label{def:G_inv_poly}
    We define the $G$-invariant maximal $\epsilon$-LDP polytope by
    \begin{align}\label{eq:G_inv_poly}
        &\mathcal{S}_{\mathcal{X},\epsilon}^G :=\left\{ w \in \mathbb{R}_+^{\mathcal{B}(\mathcal{X})/G} \middle|\; \forall \mathfrak{O} \in \mathcal{X}/G,  \sum_{\mathcal{O} \in \mathcal{B}(\mathcal{X})/G} w_{\mathcal{O}} \tilde{r}_{\mathfrak{O}, \mathcal{O}} = 1 \right\},
    \end{align}
    where 
    \begin{align}
        r_{\mathfrak{O}, \mathcal{O}} &:= |\{y \in \mathcal{O} \mid x \in y\}|,
        \\ \tilde{r}_{\mathfrak{O}, \mathcal{O}} &:= e^\epsilon r_{\mathfrak{O}, \mathcal{O}} + (|\mathcal{O}| - r_{\mathfrak{O}, \mathcal{O}}),\label{eq:G_inv_r_v}
    \end{align}
    and $x \in \mathfrak{O}$ is an arbitrary representative element of an orbit $\mathfrak{O} \in \mathcal{X}/G$.\footnote{Because each $g \in G$ acts bijectively on the orbit $\mathcal{O}$ and $x \in y$ if and only if $gx \in gy$, $r_{\mathfrak{O}, \mathcal{O}}$ is well-defined.}
    Then, we define a map $Q_{\mathrm{ex},G}^{(\cdot)}$ defined on $\mathcal{S}^G_{\mathcal{X},\epsilon}$ by
    \begin{equation}\label{eq:G_inv_extremal_LDP_ch}
        Q_{\mathrm{ex},G}^{(w)} := \bigoplus_{\mathcal{O} \in \mathcal{B}(\mathcal{X})/G} w_\mathcal{O}S_{\mathcal{X},\mathcal{O},\epsilon},
    \end{equation}
    where $S_{\mathcal{X},\mathcal{O},\epsilon} := \bigoplus_{y \in \mathcal{O}} (S_{\mathcal{X},\epsilon})_{y,*}$.
\end{definition}

\begin{theorem}\label{thm:G_inv_poly}
    The map $Q_{\mathrm{ex},G}^{(\cdot)}$ defined in \eqref{eq:G_inv_extremal_LDP_ch} establishes a one-to-one correspondence between $\mathcal{S}_{\mathcal{X},\epsilon}^G$ and $\max^G(\mathcal{Q}_{\mathcal{X},\epsilon} )/\sim$.
\end{theorem}

The proof of the above theorem is provided at the end of this section.
A direct consequence of this theorem and Theorem~\ref{thm:opt_PUT_red} is that in minimizing a $G$-invariant map $R^*$ satisfying the DPI and DS-QCVX, it is sufficient to restrict the optimization to $\mathcal{S}_{\mathcal{X},\epsilon}^G$, i.e., 
\begin{equation}\label{eq:suff_G_inv_max_red}
    \inf\limits_{Q \in \mathcal{Q}_{\mathcal{X},\epsilon}} R^*(Q) = \min\limits_{w \in \mathcal{S}^G_{\mathcal{X},\epsilon}} R^*\left(Q_{\mathrm{ex},G}^{(w)}\right).
\end{equation}

\begin{remark}
    Under the same conditions as in Remark~\ref{rmk:LP_DPI}, the above optimization problem can also be computed via vertex enumeration or linear programming.
\end{remark}

If $G$ acts transitively on $\mathcal{X}$, there is only a single input orbit, meaning $\mathcal{X}/G = \{\mathcal{X}\}$.
When combined with an additional condition that $R^*$ is concave, this fact allows us to exactly characterize the optimal weight vector, and thereby the optimal PUT, without relying on optimization solvers.
The proof of the following corollary is in Appendix~\ref{sec:PUT_transitive}.

\begin{corollary}\label{cor:PUT_transitive}
    Let $G$ be a group that acts transitively on an input alphabet $\mathcal{X}$ of size $m = |\mathcal{X}|$, and let $R^*: \mathcal{P}(\mathcal{X} \rightarrow *)\rightarrow \mathbb{R}$ be a $G$-invariant map satisfying the DPI, DS-QCVX, and CCV.
    Then, there exists an optimal LDP channel in the form of $Q = w_{\mathcal{O}} S_{\mathcal{X},\mathcal{O},\epsilon}$ for some orbit $\mathcal{O} \in \mathcal{B}(\mathcal{X})/G$, where
    \begin{equation}
        w_{\mathcal{O}} = \frac{m}{|\mathcal{O}| (k_{\mathcal{O}} e^\epsilon + m - k_{\mathcal{O}})},
    \end{equation}
    and $k_{\mathcal{O}} := |y|$ is the uniform subset size for all $y \in \mathcal{O}$.\footnote{Because $G$ acts transitively on any given orbit $\mathcal{O} \in \mathcal{B}(\mathcal{X})/G$ and $|y| = |gy|$ for all $g \in G$ and $y \in \mathcal{B}(\mathcal{X})$, $k_{\mathcal{O}}$ is well-defined.}
    Also, we have
    \begin{equation}\label{eq:PUT_transitive}
        \inf\limits_{Q \in \mathcal{Q}_{\mathcal{X},\epsilon}}R^*(Q) = \min\limits_{\mathcal{O} \in \mathcal{B}(\mathcal{X})/G} R^*\left( w_{\mathcal{O}} S_{\mathcal{X},\mathcal{O},\epsilon} \right),
    \end{equation}
    where the right-hand side is a minimization over a finite set of orbits.
\end{corollary}

\begin{remark}\label{rmk:SS}
    One of the most representative transitive group actions on a finite set $\mathcal{X}$ is the natural permutation action of the symmetric group $G=\mathrm{Sym}(\mathcal{X})$. 
    In this case, each orbit $\mathcal{O} \in \mathcal{B}(\mathcal{X})/G$ is the set of all $k_{\mathcal{O}}$-subsets of $\mathcal{X}$.
    If $R^*$ is $\mathrm{Sym}(\mathcal{X})$-invariant, then the optimal LDP channel $Q = w_\mathcal{O} S_{\mathcal{X},\mathcal{O},\epsilon}$ in the above corollary is a subset selection mechanism \cite{yeOptimalSchemesDiscrete2018}, where
    \begin{equation}
        w_{\mathcal{O}} = \frac{m}{\binom{m}{k_{\mathcal{O}}} (k_{\mathcal{O}} e^\epsilon + m - k_{\mathcal{O}})}.
    \end{equation}
\end{remark}

\begin{remark}
    Corollary~\ref{cor:PUT_transitive} also directly recovers the previous result~\cite[Thm.~1]{wangLocalDifferentialPrivate2019}, which establishes that for a uniformly distributed input $X \in \mathcal{X}$, the mutual information across an LDP channel is maximized by a subset selection mechanism. 
    Because the mutual information is $\mathrm{Sym}(\mathcal{X})$-invariant and satisfies DPI and DSA (Remark~\ref{rmk:info_meas}), its maximization follows immediately out of Corollary~\ref{cor:PUT_transitive} as a direct consequence.
\end{remark}

We conclude this section by providing the proof of Theorem~\ref{thm:G_inv_poly}.
\begin{IEEEproof}[Proof of Theorem~\ref{thm:G_inv_poly}]
For any $G$-invariant weight vector $c$ and $x \in \mathfrak{O}$, the column-stochastic constraint in \eqref{eq:extremal_LDP_ch} collapses as
\begin{align}
    1 &= \sum_{y \in \mathcal{B}(\mathcal{X})} c_y (S_{\mathcal{X},\epsilon})_{y,x}
    \\&= \sum_{\mathcal{O} \in \mathcal{B}(\mathcal{X})/G} w_{\mathcal{O}} \left( \sum_{y \in \mathcal{O}} (S_{\mathcal{X},\epsilon})_{y,x} \right)
    \\& = \sum_{\mathcal{O} \in \mathcal{B}(\mathcal{X})/G} w_{\mathcal{O}} \tilde{r}_{\mathfrak{O}, \mathcal{O}}. \label{eq:constraint_equiv}
\end{align}

We now prove injectivity.
For any given $w \in \mathcal{S}^G_{\mathcal{X},\epsilon}$, we define $c(w) \in \mathcal{S}_{\mathcal{X},\epsilon}$ by
\begin{equation}
    (c(w))_y = w_{\mathcal{O}},
\end{equation}
for every $\mathcal{O} \in \mathcal{B}(\mathcal{X})/G$ and $y \in \mathcal{O}$.
Then, it is clear that $Q_{\mathrm{ex},G}^{(w)}$ is equal to $Q_{\mathrm{ex}}^{(c(w))}$ up to a row permutation, yielding $Q_{\mathrm{ex},G}^{(w)} \sim Q_{\mathrm{ex}}^{(c(w))}$ by Lemma~\ref{lem:equiv_ch}.
Because $c: \mathcal{S}^G_{\mathcal{X},\epsilon} \rightarrow \mathcal{S}_{\mathcal{X},\epsilon}$ and $Q^{(\cdot)}_{\mathrm{ex}}$ are injective, $Q_{\mathrm{ex},G}^{(\cdot)} \sim Q^{(c(\cdot))}_{\mathrm{ex}}$ is also injective.

Now, we prove surjectivity.
Let $Q \in \max^G(\mathcal{Q}_{\mathcal{X},\epsilon})$ with output alphabet $\mathcal{Z}$ which is $G$-invariant with respect to a fixed $G$-action $\sigma$ on $\mathcal{Z}$.
By Theorem~\ref{thm:extremal_LDP_is_max}, there is a unique weight vector $c \in \mathcal{S}_{\mathcal{X},\epsilon}$ such that $Q \sim Q_{\mathrm{ex}}^{(c)}$.
First, we prove that the equivalence $Q \sim Q_{\mathrm{ex}}^{(c)}$ implies $g \circ_\sigma Q \sim g \circ Q_{\mathrm{ex}}^{(c)}$, where the $G$-action on the right-hand side is induced from the canonical $G$-action on $\mathcal{B}(\mathcal{X})$.
Because $Q \precsim Q_{\mathrm{ex}}^{(c)}$, there exists a column-stochastic matrix $W$ such that $Q = WQ_{\mathrm{ex}}^{(c)}$.
Then, we have
\begin{align}
    (g \circ_\sigma Q)_{z,x} &= Q_{\sigma_{g^{-1}}(z), g^{-1}x}
    \\& = \sum_{y \in \mathcal{B}(\mathcal{X})} W_{\sigma_{g^{-1}}(z), y} (Q_{\mathrm{ex}}^{(c)})_{y, g^{-1}x}
    \\& = \sum_{y \in \mathcal{B}(\mathcal{X})} W_{\sigma_{g^{-1}}(z), g^{-1}y} (Q_{\mathrm{ex}}^{(c)})_{g^{-1}y, g^{-1}x}
    \\& = \sum_{y \in \mathcal{B}(\mathcal{X})} (g \circ_\sigma W)_{z, y} (g\circ Q_{\mathrm{ex}}^{(c)})_{y,x},
\end{align}
where $g \circ_\sigma W$ is defined by $(g\circ_\sigma W)_{z,y}:= W_{\sigma_{g^{-1}}(z), g^{-1}y}$.
Since $g \circ_\sigma W$ is also a valid column-stochastic matrix, we have $g \circ_\sigma Q \precsim g \circ Q_{\mathrm{ex}}^{(c)}$.
By identical symmetric logic, we have $g \circ Q_{\mathrm{ex}}^{(c)} \precsim g \circ_\sigma Q$, thereby $g \circ_\sigma Q \sim g \circ Q_{\mathrm{ex}}^{(c)}$.

Next, we induce a $G$-action on $Q_{\mathrm{ex}}^{(c)}$ to that on a weight vector $c$.
By definition, we have
\begin{align}
    (g \circ Q_{\mathrm{ex}}^{(c)})_{y,x} &= (Q_{\mathrm{ex}}^{(c)})_{g^{-1}y, g^{-1}x}
    \\&= c_{g^{-1}y} (S_{\mathcal{X},\epsilon})_{g^{-1}y, g^{-1}x}
    \\&= c_{g^{-1}y} (S_{\mathcal{X},\epsilon})_{y, x},
\end{align}
where the last equality follows from \eqref{eq:staircase_mat} and $x \in y$ if and only if $g^{-1}x \in g^{-1}y$.
Therefore, we get $g \circ Q_{\mathrm{ex}}^{(c)} = Q_{\mathrm{ex}}^{(gc)}$.
Up to now, we have
\begin{equation}
    Q_{\mathrm{ex}}^{(c)} \sim Q = g \circ_\sigma Q \sim g \circ Q_{\mathrm{ex}}^{(c)} = Q_{\mathrm{ex}}^{(gc)}.
\end{equation}
Since $Q_{\mathrm{ex}}^{(c)} \sim Q_{\mathrm{ex}}^{(gc)}$, the injectivity of $Q_{\mathrm{ex}}$ requires $c = gc$ for all $g \in G$, i.e., the weight vector $c$ is $G$-invariant.
We can therefore define an orbit-weight vector $w \in \mathbb{R}_+^{\mathcal{B}(\mathcal{X})/G}$ by setting $w_{\mathcal{O}} := c_y$ for $y \in \mathcal{O}$, which exactly yields $c(w) = c$.
Finally, because $c \in \mathcal{S}_{\mathcal{X},\epsilon}$, the constraint equivalence derived in \eqref{eq:constraint_equiv} guarantees that $w \in \mathcal{S}_{\mathcal{X},\epsilon}^G$, which proves surjectivity and completes the proof.

\end{IEEEproof}

\section{Applications}\label{sec:app}

Here, we provide two applications of our results introduced in the previous section.
These examples show how our results help derive analytical expressions for PUTs.
Specifically, we consider instances of two representative statistical decision-making problems: hypothesis testing and parametric estimation.
In addition, we consider a natural action of the symmetric group $G=\mathrm{Sym}(\mathcal{X})$, while for parametric estimation, we consider a cyclic group action.

\subsection{Privacy-Preserving Hypothesis Testing}\label{sec:priv_HT}

We consider $m$-ary hypothesis testing over smoothed point mass distributions.
In this scenario, the set of model parameters coincides with the input alphabet, $\Theta = \mathcal{X} = [m]$.
We model the data-generating distribution as a one-hot profile diluted by uniform background noise:
\begin{equation}
    (P_{X|\vartheta})_{x,\theta} = \frac{1-\gamma}{m} + \gamma \delta_{x, \theta},
\end{equation}
where $\gamma \in (0,1]$, and $\delta_{x,\theta}$ denotes the Kronecker delta.
The decision space $\mathcal{A}$ is also identical to $\mathcal{X}$, and the loss function is the zero-one loss, $l(\theta,a) = 1 - \delta_{\theta, a}$.
Thus, the risk corresponds to the error probability, and minimizing the expected risk is equivalent to maximizing the guessing probability.
In the Bayesian framework, we consider the uniform prior $\lambda = \mathrm{Unif}(\mathcal{X})$.

This scenario satisfies the conditions for a $G$-invariant statistical decision-making problem under the natural action of the symmetric group $G = \mathrm{Sym}(\mathcal{X})$, where the $G$-actions on $\mathcal{X}=\Theta = \mathcal{A}$ are the natural permutation action.
Therefore, Proposition~\ref{prop:BM_structure} and Corollary~\ref{cor:PUT_transitive} imply that a subset selection (SS) mechanism introduced in Remark~\ref{rmk:SS} achieves the optimal PUT.

Now, we compute the optimal PUTs.
Let $Q = w_\mathcal{O} S_{\mathcal{X},\mathcal{O},\epsilon}$ be an SS mechanism.
Then, the Bayesian risk for a given decision rule $P_{A|Y}$ is calculated as
\begin{align}
    R_{\mathrm{B}}(Q,P_{A|Y};\lambda)  &= \frac{1}{m}\sum_{a=1}^m \sum_{y \in \mathcal{O}} P_{A|Y}(a|y)
    \left[ \sum_{\theta \neq a} \sum_{x=1}^m (P_{X|\vartheta})_{x,\theta} w_\mathcal{O}(S_{\mathcal{X},\epsilon})_{y,x} \right].
\end{align}    

Because $R_B$ is linear in $P_{A|Y}(\cdot|y)$ for all $y \in \mathcal{Y}$, the optimal Bayesian risk is given by
\begin{equation}
    R_{\mathrm{B}}^*(Q;\lambda) \\= \frac{1}{m} \sum_{y \in \mathcal{O}} \min_{a \in [m]} \left[ \sum_{\theta \neq a} \sum_{x=1}^m (P_{X|\vartheta})_{x,\theta} w_\mathcal{O}(S_{\mathcal{X},\epsilon})_{y,x} \right].
\end{equation}
Evaluating the innermost sum over $x \in [m]$ yields
\begin{equation}
    \sum_{x=1}^m \left( \frac{1-\gamma}{m} + \gamma \delta_{x, \theta} \right) w_{\mathcal{O}}(S_{\mathcal{X},\epsilon})_{y,x}
    \\ = w_{\mathcal{O}} \left( \frac{1-\gamma}{m} (k_{\mathcal{O}}e^\epsilon + m - k_{\mathcal{O}}) + \gamma (S_{\mathcal{X},\epsilon})_{y,\theta}\right).
\end{equation}
Substituting this back, the sum over $\theta \neq a$ can be computed as
\begin{equation}
     w_{\mathcal{O}} \sum_{\theta \neq a}\left( \frac{1-\gamma}{m} C + \gamma (S_{\mathcal{X},\epsilon})_{y,\theta}\right)
     = w_{\mathcal{O}} \left[ \left(\frac{(1-\gamma)(m-1)}{m}  + \gamma \right) C - \gamma (S_{\mathcal{X},\epsilon})_{y,a} \right],
\end{equation}    
where $C := k_{\mathcal{O}}e^\epsilon + m - k_{\mathcal{O}}$.

Because the above is minimized over $a \in [m]$ whenever $a \in y$, the optimal Bayes estimator $P_{A|Y}^*$ can be chosen as
\begin{equation}
    P_{A|Y}^*(a|y) = \frac{1}{k_{\mathcal{O}}} \mathbbm{1}(a \in y).
\end{equation}
Therefore,
\begin{align}
    R_{\mathrm{B}}^*(Q;\lambda) &= \frac{1}{m} \sum_{y \in \mathcal{O}}  w_{\mathcal{O}} \left[ \left(\frac{(1-\gamma)(m-1)}{m}  + \gamma \right) C - \gamma e^\epsilon \right]
    \\& = \frac{1}{C}\left[ \left(\frac{(1-\gamma)(m-1)}{m}  + \gamma \right) C - \gamma e^\epsilon \right]
    \\& = 1 - \frac{1-\gamma}{m} - \frac{\gamma e^\epsilon}{C}.
\end{align}    

Because $R_{\mathrm{B}}^*$ is increasing in $k_{\mathcal{O}} \in [m-1]$, $k_{\mathcal{O}}=1$ minimizes the risk, yielding 
\begin{equation}
    \mathrm{PUT}_{\mathrm{B}}(\mathcal{X},\epsilon;\lambda) = 1 - \frac{1-\gamma}{m} - \frac{\gamma e^\epsilon}{e^\epsilon + m - 1}.
\end{equation}

For the minimax framework, we show that $\mathrm{PUT}_{\mathrm{M}}(\mathcal{X},\epsilon)$ is equal to the Bayesian PUT derived above based on the equalizer rule \cite[Cor.~5.1.5]{TheoryPointEstimation1998}, \cite[Sec.~5.3.2]{bergerStatisticalDecisionTheory1985}.
For any given $Q \in \mathcal{Q}_{\mathcal{X},\epsilon}$, the equalizer rule states that if a decision rule $P_{A|Y}$ is the optimal Bayes decision rule for a given prior $\lambda \in \mathcal{P}(\Theta)$ such that $R(\theta,Q,P_{A|Y})$ is constant over $\theta \in \Theta$, then $P_{A|Y}$ is also the optimal minimax decision rule.
For the optimal Bayes decision rule $P_{A|Y}^*$ for a uniform distribution $\lambda = \mathrm{Unif}(\mathcal{X})$ described in the above, we have 

\begin{align}
    R(\theta,Q,P^*_{A|Y})  &= w_{\mathcal{O}}  \sum_{y \in \mathcal{O}} \left( 1 - \frac{1}{k_{\mathcal{O}}} \mathbbm{1}(\theta \in y) \right)
    \sum_{x=1}^m \left( \frac{1-\gamma}{m} + \gamma \delta_{x, \theta} \right)  (S_{\mathcal{X},\epsilon})_{y,x}
    \\& = w_{\mathcal{O}}  \sum_{y \in \mathcal{O}} \left( 1 - \frac{1}{k_{\mathcal{O}}} \mathbbm{1}(\theta \in y) \right)
    \left( \frac{1-\gamma}{m} C + \gamma (S_{\mathcal{X},\epsilon})_{y,\theta}\right)
    \\&= w_{\mathcal{O}} \Bigg[ \binom{m-1}{k_{\mathcal{O}}-1} \left( 1 - \frac{1}{k_{\mathcal{O}}} \right) \left( \frac{1-\gamma}{m} C + \gamma e^\epsilon \right)  + \binom{m-1}{k_{\mathcal{O}}}\left( \frac{1-\gamma}{m} C + \gamma \right) \Bigg].
\end{align}

Since $R(\theta,Q,P^*_{A|Y})$ does not depend on $\theta$, $R_{\mathrm{M}}^*(Q) = R_{\mathrm{B}}^*(Q;\lambda)$ for all channels in the form of $Q = w_{\mathcal{O}} S_{\mathcal{X},\mathcal{O},\epsilon}$ by the equalizer rule.
Therefore, we get
\begin{align}\label{eq:PUT_B_eq_PUT_M_1}
    \mathrm{PUT}_{\mathrm{B}}(\mathcal{X},\epsilon;\lambda) &\leq \mathrm{PUT}_{\mathrm{M}}(\mathcal{X},\epsilon) \leq \min\limits_{\mathcal{O} \in \mathcal{B}(\mathcal{X})/G}R_{\mathrm{M}}^*(w_{\mathcal{O}} S_{\mathcal{X},\mathcal{O},\epsilon}) \\& = \min\limits_{\mathcal{O} \in \mathcal{B}(\mathcal{X})/G}R_{\mathrm{B}}^*(w_{\mathcal{O}} S_{\mathcal{X},\mathcal{O},\epsilon};\lambda) = \mathrm{PUT}_{\mathrm{B}}(\mathcal{X},\epsilon;\lambda).\label{eq:PUT_B_eq_PUT_M_2}
\end{align}

\subsection{Privacy-Preserving Parametric Estimation}\label{sec:priv_PE}

We consider a parametric estimation problem where the data-generating model follows a discrete cardioid distribution
\begin{equation}
    (P_{X|\vartheta})_{x,\theta} = \frac{1}{m} \left( 1 + \gamma \cos\left( \frac{2\pi x}{m} - \theta \right) \right),
\end{equation}
where $\mathcal{X} = \{0,1,\dots,m-1\}$, $m \geq 3$, $\Theta = \mathcal{A} = [0,2\pi)$, and $\gamma \in (0, 1]$.
The loss function is defined as the cosine distance
\begin{equation}
    l(\theta,a) = 1 - \cos(\theta - a).
\end{equation}
In the Bayesian framework, we consider a uniform prior $\lambda$ on $\Theta$.

This problem forms a $G$-invariant statistical decision-making problem under the action of the cyclic group $G = \mathbb{Z}_m$, where the $G$-action on $\mathcal{X}$ is the natural cyclic shift $gx = (g + x) \pmod{m}$, and on $\Theta = \mathcal{A}$ is a rotation by discrete steps $g\theta = \left(\theta + \frac{2\pi g}{m}\right) \pmod{2\pi}$.
Because $\mathbb{Z}_m$ acts transitively on $\mathcal{X}$, Proposition~\ref{prop:BM_structure} and Corollary~\ref{cor:PUT_transitive} imply that an $\epsilon$-LDP channel of the form $Q = w_\mathcal{O} S_{\mathcal{X},\mathcal{O},\epsilon}$ achieves the optimal PUT.

Now, we compute the optimal PUTs.
Let $Q = w_\mathcal{O} S_{\mathcal{X},\mathcal{O},\epsilon}$.
Then, the Bayesian risk for a given decision rule $P_{A|Y}$ is calculated as
\begin{align}
    R_{\mathrm{B}}(Q,P_{A|Y};\lambda) &= \frac{1}{2\pi}\int_0^{2\pi} \int_0^{2\pi} \sum_{x=0}^{m-1} \sum_{y \in \mathcal{O}} P_{X|\vartheta}(x|\theta) w_\mathcal{O}(S_{\mathcal{X},\epsilon})_{y,x}
    l(\theta,a)  P_{A|Y}(da|y) d\theta 
    \\ &= \sum_{y \in \mathcal{O}} w_\mathcal{O} \int_0^{2\pi} P_{A|Y}(da|y)\left[ \frac{1}{m} \sum_{x=0}^{m-1} (S_{\mathcal{X},\epsilon})_{y,x} \int_0^{2\pi} P_{X|\vartheta}(x|\theta) l(\theta,a) \frac{d\theta}{2\pi} \right].
\end{align}

The inner integral over $\theta \in [0,2\pi)$ can be evaluated as
\begin{align}
    \int_0^{2\pi} P_{X|\vartheta}(x|\theta) l(\theta,a) \frac{d\theta}{2\pi} = 1 - \frac{\gamma}{2} \cos\left( \frac{2\pi x}{m} - a \right).
\end{align}
Substituting this back, the risk becomes
\begin{align}
    R_{\mathrm{B}}(Q,P_{A|Y};\lambda) &= \sum_{y \in \mathcal{O}} w_\mathcal{O} \int_0^{2\pi} P_{A|Y}(da|y)
    \left[ \frac{1}{m} \sum_{x=0}^{m-1} (S_{\mathcal{X},\epsilon})_{y,x} \left( 1 - \frac{\gamma}{2} \cos\left( \frac{2\pi x}{m} - a \right) \right) \right]
    \\& = \sum_{y \in \mathcal{O}} w_\mathcal{O} \int_0^{2\pi} P_{A|Y}(da|y) \left[ \frac{C}{m} - \frac{\gamma}{2m} \sum_{x=0}^{m-1} (S_{\mathcal{X},\epsilon})_{y,x} \cos\left( \frac{2\pi x}{m} - a \right) \right]
    \\& = \frac{1}{|\mathcal{O}|} \sum_{y \in \mathcal{O}} \int_0^{2\pi} P_{A|Y}(da|y)
     \left[ 1 - \frac{\gamma (e^\epsilon - 1)}{2C} \sum_{x \in y} \cos\left( \frac{2\pi x}{m} - a \right) \right],
\end{align}    
where the last equality follows from the fact that $\sum_{x=0}^{m-1} \cos(\frac{2\pi x}{m} - a) = 0$.

Because the above risk is linear in $P_{A|Y}(\cdot|y)$ for all $y \in \mathcal{Y}$, the optimal Bayesian risk is given by 
\begin{align}
    R^*_{\mathrm{B}}(Q;\lambda)
     = \frac{1}{|\mathcal{O}|} \sum_{y \in \mathcal{O}} \min\limits_{a \in [0,2\pi)} \left[ 1 - \frac{\gamma(e^\epsilon - 1)}{2C} \sum\limits_{x \in y} \cos\left( \frac{2\pi x}{m} - a \right) \right],
\end{align}
where $C = k_{\mathcal{O}}e^\epsilon + m - k_{\mathcal{O}}$.
The internal minimization over $a \in [0,2\pi)$ is equivalent to
\begin{align}
    \max_{a \in [0,2\pi)} \sum_{x \in y} \cos\left( \frac{2\pi x}{m} - a \right) & = \max_{a \in [0,2\pi)} \mathrm{Re}\left(e^{-ia} \sum_{x \in y} e^{2\pi i x / m}\right)
    \\&  = \left\| Z_y \right\|, \label{eq:pf_PE_max_a}
\end{align}
where $Z_y := \sum_{x \in y} e^{2\pi i x / m}$ and the last equality is achieved by $a = a^*(y) := \mathrm{arg}(Z_y)$ when $\|Z_y\| > 0$, and an arbitrary $a \in [0,2\pi)$ achieves the maximum value $0$ when $\|Z_y\| = 0$.
Therefore, we can specify the optimal Bayesian decision rule $P^*_{A|Y}$ by a deterministic decision rule $y \mapsto a^*(y)$ when $\|Z_y\| > 0$, and a uniform random decision rule over $[0,2\pi)$ when $\|Z_y\| = 0$.

Returning to the optimal Bayesian risk, we get
\begin{align}
    R_{\mathrm{B}}^*(Q;\lambda) &=  1- \sum_{y \in \mathcal{O}} \frac{\gamma(e^\epsilon - 1)}{2|\mathcal{O}|(k_{\mathcal{O}}e^\epsilon + m - k_{\mathcal{O}})} \left\|Z_y\right\|
    \\& = 1 - \frac{\gamma(e^\epsilon - 1)}{2(k_{\mathcal{O}}e^\epsilon + m - k_{\mathcal{O}})} \left\|Z_y\right\|,
\end{align}
where the last equality follows from the fact that the subsets $y \in \mathcal{O}$ are generated by the cyclic action of $\mathbb{Z}_m$, meaning $\left\| Z_y \right\|$ remains constant across all $y \in \mathcal{O}$.
This yields
\begin{equation}
    \mathrm{PUT}_{\mathrm{B}}(\mathcal{X},\epsilon;\lambda)
    = 1 - \max_{\mathcal{O} \in \mathcal{B}(\mathcal{X})/G} \frac{\gamma(e^\epsilon - 1)}{2(k_{\mathcal{O}}e^\epsilon + m - k_{\mathcal{O}})} \left\|Z_y\right\|,
\end{equation}
where $y \in \mathcal{O}$.

To solve the maximization over $\mathcal{O} \in \mathcal{B}(\mathcal{X})/G$, we first maximize over the orbits with a uniform subset size $k_{\mathcal{O}} = k$, and then maximize over $k \in [m-1]$.
For a fixed subset size $k_{\mathcal{O}} = k$, $\left\|Z_y\right\|$ is maximized when $y$ is a subset consisting of consecutive elements (for example, $y = \{0,1,\ldots,k-1\}$).
To see this, note that $\|Z_y\| = \max_{\theta \in [0, 2\pi)} \mathrm{Re}\left( e^{-i\theta} Z_y \right) = \max_{\theta \in [0, 2\pi)} \sum_{x \in y} \cos\left(\frac{2\pi x}{m} - \theta\right)$, and for any fixed angle $\theta$, the sum is maximized by selecting the $k$ distinct elements $x \in \{0, 1, \dots, m-1\}$ that yield the largest values of $\cos\left(\frac{2\pi x}{m} - \theta\right)$.
Hence, $y$ that maximizes $\|Z_y \|$ must consist of the $k$ values of $x$ whose corresponding angles $\frac{2\pi x}{m}$ are closest to $\theta$ on the unit circle, naturally forming a contiguous sequence.
Thus, we have
\begin{equation}
    \max_{\mathcal{O} \in \mathcal{B}(\mathcal{X})/G : k_{\mathcal{O}} = k} \left\|Z_y\right\| = \left\|\sum_{x = 0}^{k-1} e^{2\pi i x / m}\right\| = \frac{\sin(\frac{\pi k}{m})} {\sin(\frac{\pi}{m})}.
\end{equation}
Therefore,
\begin{equation}
    \mathrm{PUT}_{\mathrm{B}}(\mathcal{X},\epsilon;\lambda) = 1 - \frac{\gamma(e^\epsilon - 1)}{2\sin(\frac{\pi}{m})} \cdot \max_{k \in [m-1]}\frac{\sin(\frac{\pi k}{m})}{ke^\epsilon + m - k}.
\end{equation}

For the minimax framework, we show that $\mathrm{PUT}_{\mathrm{M}}(\mathcal{X},\epsilon)$ is equal to the Bayesian $\mathrm{PUT}_{\mathrm{B}}(\mathcal{X},\epsilon;\lambda)$ by applying the equalizer rule \cite[Cor.~5.1.5]{TheoryPointEstimation1998}, \cite[Sec.~5.3.2]{bergerStatisticalDecisionTheory1985}.
In deriving the optimal Bayesian risk for a given $Q = w_{\mathcal{O}} S_{\mathcal{X},\mathcal{O},\epsilon}$, it was shown that the optimal Bayes decision rule $P^*_{A|Y}$ for the uniform prior is the deterministic rule $y \mapsto a^*(y)$ when $\|Z_y\| > 0$, and the uniform random decision rule when $\|Z_y\| = 0$.
Because the subsets $y \in \mathcal{O}$ are generated by the cyclic action of $\mathbb{Z}_m$, the magnitude $\|Z_y\|$ is constant over $y \in \mathcal{O}$. 
Therefore, we evaluate the risk $R(\theta,Q,P^*_{A|Y})$ by considering two disjoint cases for the entire orbit: $\|Z_y\| = 0$ and $\|Z_y\| > 0$.

First, assume $\|Z_y\| = 0$.
In this case, the optimal Bayes estimator is the uniform randomized rule $P^*_{A|Y}(da|y) = \frac{da}{2\pi}$, and we have
\begin{align}
    R(\theta,Q,P^*_{A|Y}) & = \frac{1}{2\pi}\int_0^{2 \pi} \sum_{x=0}^{m-1} \sum_{y \in \mathcal{O}} P_{X|\vartheta}(x|\theta) w_\mathcal{O}(S_{\mathcal{X},\epsilon})_{y,x}l(\theta,a) da
    \\& = \left[\frac{1}{2\pi}\int_0^{2\pi} \left( 1 - \cos(\theta - a) \right) da \right]
    \sum_{y \in \mathcal{O}}\sum_{x=0}^{m-1} \frac{{w_\mathcal{O}} }{m}\left( 1 + \gamma \cos\left( \frac{2\pi x}{m} - \theta \right) \right) (S_{\mathcal{X},\epsilon})_{y,x}
    \\& = 1 + \frac{\gamma(e^\epsilon - 1)}{|\mathcal{O}|(k_{\mathcal{O}}e^\epsilon + m - k_{\mathcal{O}})} \sum\limits_{y \in \mathcal{O}}\sum_{x \in y} \cos\left( \frac{2\pi x}{m} - \theta \right)
    \\& = 1 + \frac{\gamma(e^\epsilon - 1)}{|\mathcal{O}|(k_{\mathcal{O}}e^\epsilon + m - k_{\mathcal{O}})} \sum\limits_{y \in \mathcal{O}}\mathrm{Re}(e^{-i \theta} Z_y) = 1.
\end{align}

Now, assume $\|Z_y\| > 0$.
In this case, the optimal Bayes estimator is the deterministic decision rule corresponding to $a^*(y) = \arg (Z_y)$, and we have
\begin{align}
    R(\theta,Q,P^*_{A|Y}) & = \sum_{x=0}^{m-1} \sum_{y \in \mathcal{O}} P_{X|\vartheta}(x|\theta) w_\mathcal{O}(S_{\mathcal{X},\epsilon})_{y,x}l(\theta,a^*(y))
    \\& = {w_\mathcal{O}} \sum_{y \in \mathcal{O}} \left( 1 - \cos(\theta - a^*(y)) \right)
    \sum_{x=0}^{m-1} \frac{1}{m}\left( 1 + \gamma \cos\left( \frac{2\pi x}{m} - \theta \right) \right) (S_{\mathcal{X},\epsilon})_{y,x}
    \\& = \frac{1}{|\mathcal{O}|}\sum_{y \in \mathcal{O}} \left( 1 - \cos(\theta - a^*(y)) \right)
    \left( 1 + \frac{\gamma(e^\epsilon - 1)}{k_{\mathcal{O}}e^\epsilon + m - k_{\mathcal{O}}} \sum_{x \in y} \cos\left( \frac{2\pi x}{m} - \theta \right)\right)
    \\& = \frac{1}{|\mathcal{O}|}\sum_{y \in \mathcal{O}} \left( 1 - \cos(\theta - a^*(y)) \right)
    \left( 1 + \frac{\|Z_y\|\gamma(e^\epsilon - 1)}{k_{\mathcal{O}}e^\epsilon + m - k_{\mathcal{O}}} \cos(\theta - a^*(y))\right),
\end{align}
where the last equality follows from 
\begin{align}
     \sum_{x \in y} \cos\left( \frac{2\pi x}{m} - \theta \right) &= \mathrm{Re} (e^{-i\theta} Z_y)
     \\& = \|Z_y \| \cos( a^*(y) - \theta).
\end{align}
Therefore, the risk $R(\theta,Q,P^*_{A|Y})$ is a quadratic polynomial in $\cos( \theta - a^*(y))$:
\begin{align}
    R(\theta,Q,P^*_{A|Y}) &= 1 + \frac{L_{\mathcal{O}} - 1}{|\mathcal{O}|} \sum_{y \in \mathcal{O}}  \cos(a^*(y) - \theta) 
    - \frac{L_{\mathcal{O}}}{|\mathcal{O}|} \sum\limits_{y \in \mathcal{O}}\cos^2(a^*(y) - \theta),
\end{align}
where we define the constant
\begin{equation}
    L_{\mathcal{O}} = \frac{\gamma (e^\epsilon - 1) \|Z_y\|}{k_{\mathcal{O}}e^\epsilon + m - k_{\mathcal{O}}}.
\end{equation}

Because $\mathcal{O}$ is generated by $\mathbb{Z}_m$, $\|Z_y\|$ and $L_{\mathcal{O}}$ are strictly constant over $y \in \mathcal{O}$.
Furthermore, if $|\mathcal{O}| \leq 2$, the only possible non-empty subsets $y \in \mathcal{O}$ are $\mathcal{X}$, the set of all even elements in $\mathcal{X}$, or the set of all odd elements in $\mathcal{X}$.
For such cases, the symmetric configuration forces $\|Z_y \| = 0$.
Thus, the condition $\|Z_y\| > 0$ ensures that $|\mathcal{O}| \geq 3$.
Additionally, we have $a^*(g y) = a^*(y) + \frac{2\pi g}{m} \pmod{2\pi}$, meaning the angles $\{a^*(y)\}_{y \in \mathcal{O}}$ uniformly partition the unit circle.
Therefore, we have the following identities:
\begin{equation}
    \sum_{y \in \mathcal{O}} \cos(a^*(y) - \theta) = 0,
\end{equation}

\begin{equation}
    \sum_{y \in \mathcal{O}} \cos^2(a^*(y) - \theta)
    = \frac{|\mathcal{O}|}{2} + \frac{1}{2}\sum_{y \in \mathcal{O}} \cos(2a^*(y) - 2\theta) = |\mathcal{O}|/2.
\end{equation}
Substituting these identities back into the risk, we obtain
\begin{equation}
    R(\theta,Q,P^*_{A|Y}) = 1 - \frac{L_{\mathcal{O}}}{2}.
\end{equation}

Since $L_{\mathcal{O}} = 0$ when $\|Z_y\| = 0$, both disjoint cases $\|Z_y\| = 0$ and $\|Z_y\| \neq 0$ coincide with the above expression, which is independent of $\theta$.
Therefore, the equalizer rule implies that $R^*_\mathrm{M}(Q) = R^*_\mathrm{B}(Q;\lambda)$ for all channels in the form of $Q = w_{\mathcal{O}} S_{\mathcal{X},\mathcal{O},\epsilon}$.
From \eqref{eq:PUT_B_eq_PUT_M_1}--\eqref{eq:PUT_B_eq_PUT_M_2}, we conclude that $\mathrm{PUT}_\mathrm{M}(\mathcal{X},\epsilon) = \mathrm{PUT}_\mathrm{B}(\mathcal{X},\epsilon;\lambda)$.

\section{Conclusion and Future Directions}\label{sec:conc}

In this work, we introduced a general scenario for privacy-preserving statistical decision-making and analyzed the optimal PUT and the corresponding optimal LDP channels in a unified manner.
These results were derived by jointly considering the functional properties of optimal risks (DPI, DS-QCVX/DSA, CCV, and group invariance) alongside the intrinsic geometry of the LDP constraint space.
We also provide applications of these results, yielding analytic characterizations for the optimal PUTs and the corresponding optimal LDP channels that were previously unaddressed.

We conclude by introducing several potential directions for future research.
\begin{itemize}
    \item \textit{Distributed Setting}: Privacy-preserving statistical decision-making scenario introduced in Section~\ref{sec:stat_dec} could naturally be extended to distributed scenarios involving $n$ data providers, where the input alphabet is defined as $\mathcal{X} = \prod_{i=1}^n \mathcal{D}_i$ and the joint channel is a product of LDP channels, $Q = \bigotimes_{i=1}^n Q^{(i)}$ with $Q^{(i)} \in \mathcal{Q}_{\mathcal{D}_i,\epsilon}$.
    While the optimal risk in this setting retains the DPI, it is not immediately clear whether properties like DS-QCVX or $G$-invariance are preserved.
    However, given that SS is optimal in both the one-shot setting (Remark~\ref{rmk:SS}) and certain asymptotic distributed tasks \cite{yeOptimalSchemesDiscrete2018, namQuantumAdvantagePrivate2025}, identifying the specific conditions under which these functional properties are preserved in distributed settings presents an intriguing direction.
    
    \item \textit{Global DP}: The geometric analysis we applied to LDP might also offer insights into global DP.
    As mentioned in Remark~\ref{rmk:general_const}, Proposition~\ref{prop:max_LDP_geo} can be extended to metric DP, which includes global DP as a special case.
    More precisely, the $\epsilon$-DP constraint can be characterized by a polyhedral cone defined by the constraint matrix $\tilde{T}_{\mathcal{X},\epsilon} \oplus I_{\mathcal{X} \times \mathcal{X}}$, where $\tilde{T}_{\mathcal{X},\epsilon}$ is a submatrix of $T_{\mathcal{X},\epsilon}$ in \eqref{eq:LDP_const_mat} consisting of rows with $z \leftrightarrow z'$, where $\leftrightarrow$ denotes a neighboring relation.
    However, the remaining task is to explicitly characterize the extreme directions of this cone, which is analogous to Lemma~\ref{lem:ext_ray_LDP}.
    Once such a characterization is established, it would be useful for calculating the optimal PUT under the global DP setup.

    \item \textit{Approximate LDP \& Infinite Alphabets}: Approximate LDP, $(\epsilon,\delta)$-LDP, is a generalization of $\epsilon$-LDP defined by the condition $\sum_{y \in B} (e^\epsilon Q_{y,x} - Q_{y,x'}) \geq -\delta$ for all $x,x' \in \mathcal{X}$ and subsets $B \subset \mathcal{Y}$.
    In the $\epsilon$-LDP setup (which is equivalent to $(\epsilon,0)$-LDP), our geometric approach exploits the fact that the constraints can be decoupled in a row-wise manner.
    However, because $(\epsilon,\delta)$-LDP constraints couple probabilities across the entire channel matrix, our current row-wise approach cannot be directly applied.
    It would be interesting to explore whether a natural matrix-level extension of our approach exists for the $(\epsilon,\delta)$-LDP setup.
    Finally, generalizing the overall framework from the finite alphabets considered in this work to accommodate arbitrary infinite alphabets remains a promising direction.
\end{itemize}

\appendix

\section{Proofs}\label{sec:pfs}

\subsection{Proof of Proposition~\ref{prop:BM_structure}}\label{sec:pf_BM_structure}

We first show that both $R_{\mathrm{B}}^*$ and $R_{\mathrm{M}}^*$ satisfies the DPI.
Let $Q \in \mathcal{P}(\mathcal{X} \rightarrow \mathcal{Y})$ and $\tilde{Q} \in \mathcal{P}(\mathcal{X} \rightarrow \mathcal{Z})$ be two channels such that $\tilde{Q} = WQ$ holds for some post-processing channel $W \in \mathcal{P}(\mathcal{Y} \rightarrow \mathcal{Z})$.
Then, for any $\theta \in \Theta$ and a decision rule $P_{A|Z}$ for $\tilde{Q}$, we have
\begin{equation}
    R(\theta,\tilde{Q},P_{A|Z}) = R(\theta,Q,P_{A|Z}W),
\end{equation} 
where $P_{A|Z}W \in \mathcal{P}(\mathcal{Y} \rightarrow \mathcal{A})$ is a valid decision rule for $Q$.
Therefore, both the optimal Bayesian risk and the minimax risk achievable by $\tilde{Q}$ are also achievable by $Q$, and thereby $R_{\mathrm{B}}^*(Q;\lambda) \leq R_{\mathrm{B}}^*(\tilde{Q};\lambda)$ and $R_{\mathrm{M}}^*(Q) \leq R_{\mathrm{M}}^*(\tilde{Q})$.

Now, we prove that $R_{\mathrm{B}}^*$ satisfies the DSA.
Let $Q^{(j)} \in \mathcal{P}(\mathcal{X} \rightarrow \mathcal{Y}^{(j)})$, $Q = \oplus_j p_j Q^{(j)} \in \mathcal{P}(\mathcal{X} \rightarrow \sqcup_j \mathcal{Y}^{(j)})$, and $\mathcal{Y} = \sqcup_j \mathcal{Y}^{(j)}$.
Any decision rule $P_{A|Y} \in \mathcal{P}(\mathcal{Y} \rightarrow \mathcal{A})$ can be naturally identified by a collection of $P_{A|Y^{(j)}} \in \mathcal{P}(\mathcal{Y}^{(j)} \rightarrow \mathcal{A})$.
With this identification, we have 
\begin{align}
    R(\theta,Q,P_{A|Y}) &= \sum_j p_j \sum\limits_{x \in \mathcal{X}, y \in \mathcal{Y}^{(j)}}  P_{X|\vartheta}(x|\theta) Q^{(j)}_{y,x}
    \int_\mathcal{A} l(\theta,a)  P_{A|Y^{(j)}}(da|y)
    \\& = \sum_j p_j R(\theta,Q^{(j)},P_{A|Y^{(j)}}), \label{eq:risk_decomp}
\end{align}
and thereby $R_{\mathrm{B}}(Q,P_{A|Y};\lambda) = \sum_j p_j R_{\mathrm{B}}(Q^{(j)},P_{A|Y^{(j)}};\lambda)$.
Because taking infimum over $P_{A|Y} \in \mathcal{P}(\mathcal{Y} \rightarrow \mathcal{A})$ is equivalent to taking infimum over $P_{A|Y^{(j)}} \in \mathcal{P}(\mathcal{Y}^{(j)}\rightarrow \mathcal{A})$ for each $j$ independently, $R^*_{\mathrm{B}}$ satisfies the DSA.

To show the DS-QCVX of $R_{\mathrm{M}}^*$, let $Q^{(j)} \in \mathcal{P}(\mathcal{X} \rightarrow \mathcal{Y}^{(j)})$, $Q = \oplus_j p_j Q^{(j)} \in \mathcal{P}(\mathcal{X} \rightarrow \sqcup_j \mathcal{Y}^{(j)})$, $\mathcal{Y} = \sqcup_j \mathcal{Y}^{(j)}$, $P_{A|Y^{(j)}}$ be an optimal decision rule that achieves $R_{\mathrm{M}}^*(Q^{(j)})$, and $P^{*}_{A|Y}\in \mathcal{P}(\mathcal{Y} \rightarrow \mathcal{A})$ be a decision rule for $Q$ such that $P^{*}_{A|Y}(\cdot|y) = P_{A|Y^{(j)}}(\cdot|y)$ if $y \in \mathcal{Y}^{(j)}$.\footnote{In general, an optimal decision rule achieving the infimum might not exist. In such cases, the proof follows identically by considering a sequence of $\delta$-approximate optimal decision rules and taking the limit as $\delta \to 0$.}
From \eqref{eq:risk_decomp}, we have 
\begin{align}
    R(\theta,Q,P^{*}_{A|Y}) &= \sum_j p_j R(\theta,Q^{(j)},P_{A|Y^{(j)}})
    \\& \leq \max_j R(\theta,Q^{(j)},P_{A|Y^{(j)}}),
\end{align}
for all $\theta \in \Theta$.
By taking supremum over $\theta \in \Theta$ on the above, we get
\begin{align}
    R_{\mathrm{M}}^*(Q) &\leq \sup_{\theta \in \Theta} R(\theta,Q,P^{*}_{A|Y})
    \\& \leq \max_j \sup_{\theta \in \Theta} R(\theta,Q^{(j)},P_{A|Y^{(j)}})
    \\& = \max_j R_{\mathrm{M}}^*(Q^{(j)}).
\end{align}

Next, we prove the $G$-invariance of $R_{\mathrm{B}}^*$ and $R_{\mathrm{M}}^*$.
As we defined the induced $G$-action on $Q$ in Definition~\ref{def:invariant_ch}, we can similarly define the induced $G$-action on a decision rule $P_{A|Y}$ by 
\begin{equation}
    (g \circ_\sigma P_{A|Y}) (\cdot | y) = P_{A|Y}( g^{-1}(\cdot) \mid \sigma_{g^{-1}} (y)).
\end{equation}
Then, for all $g \in G$, $\theta \in \Theta$, finite set $\mathcal{Y}$, $Q \in \mathcal{P}(\mathcal{X} \rightarrow \mathcal{Y})$, $G$-action $\sigma$ on $\mathcal{Y}$, and $P_{A|Y} \in \mathcal{P}(\mathcal{Y} \rightarrow \mathcal{A})$, we have 
\begin{align}
    R(\theta,g \circ_\sigma Q , g\circ_\sigma P_{A|Y}) &= \sum\limits_{x \in \mathcal{X}, y \in \mathcal{Y}}  P_{X|\vartheta}(x|\theta) Q_{\sigma_{g^{-1}} (y),g^{-1}x}
    \int_\mathcal{A} l(\theta,a)  P_{A|Y}(g^{-1}(da)|\sigma_{g^{-1}} (y))
    \\& = \sum\limits_{x \in \mathcal{X}, y \in \mathcal{Y}}  P_{X|\vartheta}(gx|\theta) Q_{y,x} \int_\mathcal{A} l(\theta,ga)  P_{A|Y}(da|y)
    \\& = \sum\limits_{x \in \mathcal{X}, y \in \mathcal{Y}}  P_{X|\vartheta}(x|g^{-1}\theta) Q_{y,x} \int_\mathcal{A} l(g^{-1}\theta,a)  P_{A|Y}(da|y)
    \\& = R(g^{-1}\theta, Q , P_{A|Y}).
\end{align}
By taking supremum over $\theta$ and using the fact that $\theta \mapsto g^{-1}\theta$ is a bijection, we have $R^*_{\mathrm{M}}(g \circ_\sigma Q) = R^*_{\mathrm{M}}(Q)$.
Also, by the $G$-invariance of the prior $\lambda \in \mathcal{P}(\Theta)$, we have $R_{\mathrm{B}}^*(g \circ_\sigma Q;\lambda) = R_{\mathrm{B}}^*(Q;\lambda)$. \hfill $\blacksquare$

\subsection{Proof of Theorem~\ref{thm:opt_PUT_red}} \label{sec:pf_opt_PUT_red}

The proof relies on geometric properties of $\mathcal{Q}_{\mathcal{X},\epsilon}$ and $\max(\mathcal{Q}_{\mathcal{X},\epsilon})$ that we formally established in Section~\ref{sec:max_LDP}.
By Proposition~\ref{prop:max_LDP_geo}, every LDP channel is dominated by a maximal LDP channel.
Therefore, if $R^*$ satisfies the DPI, then it is clear from the definitions that 
\begin{equation}
    \inf\limits_{Q \in \mathcal{Q}_{\mathcal{X},\epsilon}} R^*(Q) = \inf\limits_{Q \in \max(\mathcal{Q}_{\mathcal{X},\epsilon})} R^*(Q).
\end{equation}
Now, we prove \eqref{eq:suff_G_inv_max} by showing that for any given $Q \in \max(\mathcal{Q}_{\mathcal{X},\epsilon})$, there exists a $G$-invariant maximal $\epsilon$-LDP channel $\tilde{Q} \in \max^G(\mathcal{Q}_{\mathcal{X},\epsilon})$ such that $R^*(Q) \geq R^*(\tilde{Q})$.

Let $Q \in \max(\mathcal{Q}_{\mathcal{X},\epsilon})$ with output alphabet $\mathcal{Y}$ and we introduce the trivial $G$-action on $\mathcal{Y}$, i.e., $\sigma_g(y) = y$.
We then construct the symmetrized channel $\tilde{Q} \in \mathcal{P}(\mathcal{X}\rightarrow \tilde{\mathcal{Y}})$, defined by
\begin{equation}
    \tilde{Q} := \bigoplus_{g \in G} \frac{1}{|G|} (g\circ_\sigma Q),
\end{equation}
where $\tilde{\mathcal{Y}} = \bigsqcup_{g \in G} \mathcal{Y} = \mathcal{Y} \times G$.
By construction, the components of $\tilde{Q}$ are given by
\begin{equation}
    \tilde{Q}_{(y,g),x} = (g \circ_\sigma Q)_{y,x}/|G| = Q_{y,g^{-1}x} / |G|.    
\end{equation}
A $G$-action $\tilde{\sigma}$ on $\tilde{\mathcal{Y}}$ is naturally induced from the $G$-action $\sigma$ on $\mathcal{Y}$ by 
\begin{equation}
    \forall h,g \in G,\; y \in \mathcal{Y},\quad
    \tilde{\sigma}_h(y,g) := (\sigma_h(y),hg) = (y,hg).
\end{equation}
With respect to this action, $\tilde{Q}$ is $G$-invariant because
\begin{align}
    (h \circ_{\tilde{\sigma}} \tilde{Q})_{(y,g),x} &= \tilde{Q}_{(y,h^{-1}g),h^{-1}x} = \frac{1}{|G|}Q_{y,g^{-1}x} \\& = \frac{1}{|G|}(g\circ_\sigma Q)_{y,x} = \tilde{Q}_{(y,g),x}.
\end{align}
Furthermore, $R^*(Q) \geq R^*(\tilde{Q})$ holds because
\begin{align}
    R^*(\tilde{Q}) &= R^*\left(\bigoplus_{g \in G} \frac{1}{|G|} (g\circ_\sigma Q)\right)
    \\& \leq \max_{g \in G} R^* (g\circ_\sigma Q) = R^* (Q),
\end{align}
where the inequality follows from the DS-QCVX, and the last equality follows from the $G$-invariance of $R^*$.

Finally, it remains to prove that $\tilde{Q} \in \max(\mathcal{Q}_{\mathcal{X},\epsilon})$.
By Proposition~\ref{prop:max_LDP_geo}, every non-zero row of $Q$ corresponds to an extreme direction of the cone $\mathcal{K}_{\mathcal{X},\epsilon}$.
As constructed, each row of $\tilde{Q}$ is a permuted version of a row of $Q$.
Lemma~\ref{lem:ext_ray_LDP} guarantees that permuting an extreme direction of $\mathcal{K}_{\mathcal{X},\epsilon}$ yields another extreme direction.
Therefore, every non-zero row of $\tilde{Q}$ is also an extreme direction, and hence Proposition~\ref{prop:max_LDP_geo} implies $\tilde{Q} \in \max(\mathcal{Q}_{\mathcal{X},\epsilon})$. 
This completes the proof of \eqref{eq:suff_G_inv_max}. \hfill $\blacksquare$

\subsection{Proof of Corollary~\ref{cor:PUT_transitive}}\label{sec:PUT_transitive}
By \eqref{eq:suff_G_inv_max_red}, we can restrict the domain of optimization to $\mathcal{S}_{\mathcal{X},\epsilon}^G$.
Because $G$ acts transitively on $\mathcal{X}$, we have $\mathcal{X}/G = \{\mathcal{X}\}$.
Consequently, the column-stochastic constraint in \eqref{eq:G_inv_poly} collapses into a single linear equality constraint
\begin{equation}\label{eq:transitive_constraint}
    \sum_{\mathcal{O} \in \mathcal{B}(\mathcal{X})/G} w_{\mathcal{O}} \tilde{r}_{\mathcal{X}, \mathcal{O}} = 1.
\end{equation}

We first derive the explicit form of $\tilde{r}_{\mathcal{X}, \mathcal{O}}$.
By a standard double-counting argument, we have
\begin{equation}
   |\mathcal{X}|r_{\mathcal{X},\mathcal{O}} = \sum_{y \in \mathcal{O}}\sum_{x \in \mathcal{X}} \mathbbm{1}(x \in y) = \sum_{y \in \mathcal{O}} |y| = |\mathcal{O}| k_{\mathcal{O}}.
\end{equation}
Combining the above with \eqref{eq:G_inv_r_v}, we get 
\begin{equation}
    \tilde{r}_{\mathcal{X}, \mathcal{O}} = \frac{|\mathcal{O}|}{m} \left( k_{\mathcal{O}} e^\epsilon + m - k_{\mathcal{O}} \right).
\end{equation}

Note that the minimum of a concave function over a polytope is achieved at an extreme point of the polytope \cite[Cor.~32.4]{rockafellarConvexAnalysis1997}.
Moreover, if a polytope is defined by the positivity constraint and $n$ linear equality constraints, then at most $n$ entries of an extreme point can be non-zero \cite[Prop.~4.2]{bertsimasIntroductionLinearOptimization1997}.
Since there is only one linear equality constraint defining $\mathcal{S}_{\mathcal{X},\epsilon}^G$, every extreme point $w$ of it has exactly one non-zero entry on a single orbit $\mathcal{O} \in \mathcal{B}(\mathcal{X})/G$.
Based on the derived expression for $\tilde{r}_{\mathcal{X}, \mathcal{O}}$ and the equality constraint, we have
\begin{equation}
    w_{\mathcal{O}} = \frac{1}{\tilde{r}_{\mathcal{X}, \mathcal{O}}} = \frac{m}{|\mathcal{O}| (k_{\mathcal{O}} e^\epsilon + m - k_{\mathcal{O}})}.
\end{equation}
This completes the proof.  \hfill $\blacksquare$

\section*{Acknowledgments}
During the preparation of this manuscript, the authors used OpenAI's ChatGPT (GPT-5.6) as an assistive tool for literature search, checking for potential errors in mathematical derivations and proofs, and refining English exposition.
In addition, based on the ideas provided by the authors, the tool assisted in developing the detailed parametric estimation scenario with cyclic shift symmetry presented in Section~\ref{sec:priv_PE}, and aided in the algebraic manipulations involved in calculating the corresponding PUTs. 
Google's Gemini 3.5 Flash-Lite was also employed for enhancing the clarity of the English text. 
The underlying theoretical frameworks, main proof strategies, and final mathematical derivations were independently developed and rigorously verified by the authors, who take full responsibility for the scientific integrity of this work.

\bibliographystyle{IEEEtran}
\bibliography{ref}

\vspace{-20pt}

\begin{IEEEbiographynophoto}{Seung-Hyun Nam} (Member) received the B.S. and M.S. degrees in Electrical Engineering from Pohang University of Science and Technology (POSTECH), Pohang, South Korea, in 2018 and 2020, respectively, and Ph.D. degree in Electrical Engineering from Korea Advanced Institute of Science and Technology (KAIST), Daejeon, South Korea, in 2024.
He is currently a post-doctoral researcher with Faculty of Information Technology, University of Jyväskylä, Jyväskylä, Finland.
He was a post-doctoral researcher with Information \& Electronics Research Institute, KAIST, from 2024 to 2026.
His research interests include (both classical and quantum) information theory, statistical inference, differential privacy, and information theoretic security.
\end{IEEEbiographynophoto}

\vspace{-20pt}

\begin{IEEEbiographynophoto}{Hyun-Young Park} (Member) received the B.S. (valedictorian) and Ph.D. degrees in Electrical Engineering from Korea Advanced Institute of Science and Technology (KAIST), Daejeon, South Korea, in 2021 and 2026, respectively.
He is currently a post-doctoral researcher with the School of Electrical Engineering, KAIST. His research interests include information theory, differential privacy, and quantum information theory.
\end{IEEEbiographynophoto}

\vspace{-20pt}

\begin{IEEEbiographynophoto}{Si-Hyeon Lee} (Senior Member) received the B.S. (summa cum laude) and Ph.D. degrees in electrical engineering from the Korea Advanced Institute of Science and Technology (KAIST), Daejeon, South Korea, in 2007 and 2013, respectively. She is currently an Associate Professor with the School of Electrical Engineering, KAIST. She was a Postdoctoral Fellow with the Department of Electrical and Computer Engineering, University of Toronto, Toronto, Canada, from 2014 to 2016, and an Assistant Professor with the Department of Electrical Engineering, Pohang University of Science and Technology (POSTECH), Pohang, South Korea, from 2017 to 2020. Her research interests include information theory, wireless communications, statistical inference, and machine learning. 
She was an IEEE Information Theory Society Distinguished Lecturer (2024-2025), a TPC Co-Chair of IEEE Information Theory Workshop 2024, and  a Guest Editor for {\scshape IEEE Journal on Selected Areas in Communications} (Special issue on Secure Communication, Sensing, and Computation in Future Intelligent Wireless Networks). She is currently an Associate Editor for {\scshape IEEE Transactions on Information Theory} and {\scshape IEEE Transactions on Communications}.
\end{IEEEbiographynophoto}

\vfill

\end{document}